\documentclass{article}
\usepackage{amssymb}
\usepackage{multicol,multienum}
\usepackage{latexsym,amssymb,amsmath}
\usepackage{indentfirst}
\usepackage{cases}
\usepackage{color}
\usepackage{graphicx}
\usepackage{mathrsfs}
\usepackage{amsthm}
\usepackage{bm}
\usepackage{CJK}
\usepackage{url}
\usepackage[font={sf,bf},textfont=md]{caption}
\usepackage{animate}
\usepackage{hyperref}
\usepackage{multirow}
\usepackage{interval}
\usepackage[table,xcdraw]{xcolor}
\usepackage{subcaption}
\usepackage{caption}
\usepackage{bm}
\usepackage{thmtools, thm-restate}
\usepackage{bbm}
\usepackage{verbatim}
\usepackage{wrapfig}
\usepackage{lscape}
\usepackage{rotating}
\usepackage{algorithm}
\usepackage{enumitem}
\usepackage{natbib}
\usepackage[noend]{algpseudocode}
\usepackage{bbding}
\usepackage{booktabs}
\usepackage{titlesec}
\usepackage{bigstrut}
\usepackage[stable]{footmisc}
\usepackage{multibib}
\newcites{online}{References}

\newcommand{\beginsupplement}{%
        \setcounter{table}{0}
        \renewcommand{\thetable}{S\arabic{table}}%
        \setcounter{figure}{0}
        \renewcommand{\thefigure}{S\arabic{figure}}%
        \setcounter{page}{0}
     }

\newcommand{\pkg}[1]{{\fontseries{b}\selectfont #1}} 
\newtheorem{theorem}{Theorem}
\newtheorem{corollary}{Corollary}[theorem]
\newtheorem{lemma}{Lemma}

\DeclareMathOperator*{\argmin}{arg\,min}

\begin{document}

\def\spacingset#1{\renewcommand{\baselinestretch}%
{#1}\small\normalsize} \spacingset{1}

\title{\bf On the Use of Information Criteria for Subset Selection in Least Squares Regression}
\author{Sen Tian\footnote{E-mail: st1864@stern.nyu.edu} \quad Clifford M. Hurvich  \quad Jeffrey S. Simonoff \\\\
  Department of Technology, Operations, and Statistics, \\Stern School of Business, New York University.}
\date{}
\maketitle

\begin{abstract}
  Least squares (LS)-based subset selection methods are popular in linear regression modeling. Best subset selection (BS) is known to be NP-hard and has a computational cost that grows exponentially with the number of predictors. Recently, \citet{Bertsimas2016} formulated BS as a mixed integer optimization (MIO) problem and largely reduced the computation overhead by using a well-developed optimization solver, but the current methodology is not scalable to very large datasets. In this paper, we propose a novel LS-based method, the best orthogonalized subset selection (BOSS) method, which performs BS upon an orthogonalized basis of ordered predictors and scales easily to large problem sizes. Another challenge in applying LS-based methods in practice is the selection rule to choose the optimal subset size $k$. Cross-validation (CV) requires fitting a procedure multiple times, and results in a selected $k$ that is random across repeated application to the same dataset. Compared to CV, information criteria only require fitting a procedure once, but they require knowledge of the effective degrees of freedom for the fitting procedure, which is generally not available analytically for complex methods. Since BOSS uses orthogonalized predictors, we first explore a connection for orthogonal non-random predictors between BS and its Lagrangian formulation (i.e., minimization of the residual sum of squares plus the product of a regularization parameter and $k$), and based on this connection propose a heuristic degrees of freedom (hdf) for BOSS that can be estimated via an analytically-based expression. We show in both simulations and real data analysis that BOSS using a proposed Kullback-Leibler based information criterion AICc-hdf has the strongest performance of all of the LS-based methods considered and is competitive with regularization methods, with the computational effort of a single ordinary LS fit. Supplementary materials are attached at the end of the main document. An R package \pkg{BOSSreg}, the computer code to reproduce the results for this article, and the complete set of simulation results are available online\footnote{\url{https://github.com/sentian/BOSSreg}. A stable version of the R package is available on \textit{CRAN}.}.

\end{abstract}

\noindent%


\section{Introduction}
Suppose that we have the data generating process
\begin{equation}
\mathbf{y}=\mathbf{\mu}+\mathbf{\epsilon},
\label{eq:truemodel_def}
\end{equation}
where $\mathbf{y}\in \mathcal{R}^n$ is the response vector, $\mathbf{\mu} \in \mathcal{R}^n$ is the fixed mean vector, and $\mathbf{\epsilon} \in \mathcal{R}^n$ is the noise vector. The mean vector is estimated based on a fixed design matrix $\mathbf{X}\in \mathcal{R}^{n\times p}$. We assume the error $\mathbf{\epsilon} \sim \mathcal{N}(0,\sigma^2 I)$.

\subsection{Best subset selection}
Best subset selection (BS) \citep{Hocking1967} seeks the set of predictors that best fit the data in terms of quadratic error for each given subset size $k = 0,1,\cdots,K$ and $K = \min\{n,p\}$, i.e. it solves the following constrained optimization problem:
\begin{equation}
\min_{\beta_0, \beta} \frac{1}{2} \lVert \mathbf{y}-\beta_0\mathtt{1} -\mathbf{X\beta}\rVert_2^2 \quad \text{subject to} \quad \lVert \mathbf{\beta} \rVert_0 \le k,
\label{eq:bestsubset-setup}
\end{equation}
where $\lVert \mathbf{\beta} \rVert_0 = \sum_{i=1}^{p} \mathtt{1}(\beta_i \ne 0)$ is the number of non-zero coefficients in $\mathbf{\beta}$. Note that to simplify the discussion, we assume that the intercept term $\beta_0=0$ throughout the paper, except in the real data examples, where all of the fitting procedures include an intercept.

BS is known to be an NP-hard problem \citep{Natarajan1995}, and its computational cost grows exponentially with the dimension $p$. Many attempts have been made to reduce the computational cost of the method. The most well-known approach is the branch-and-bound algorithm ``leaps'' \citep{Furnival1974} that solves \eqref{eq:bestsubset-setup} in seconds for $p$ being up to around $30$. More recently, \citet{Bertsimas2016} formulated \eqref{eq:bestsubset-setup} using a mixed integer operator (MIO), and largely reduced the computing overhead by using a well-developed optimization solver such as {\tt GUROBI} or {\tt CPLEX}. However, according to \citet{Hastie2017}, MIO normally takes about $3$ minutes to find the solution at a given size $k$ for a problem with $n=500$ and $p=100$. The current methodology is not scalable to very large datasets, and solving \eqref{eq:bestsubset-setup} remains a challenge for most real world applications.

In order to select the optimal tuning parameter, e.g. the subset size $k$ in \eqref{eq:bestsubset-setup}, one often applies either cross-validation (CV) or an information criterion. CV requires fitting the modeling procedure multiple times, and results in a selected value that is random across repeated application to the same data. Information criteria avoid these difficulties by augmenting the training error with a function of the effective degrees of freedom (edf) that characterizes the model complexity. \citet{Efron1986} defined the edf for a general fitting rule $\hat{\mu}$ as:
\begin{equation}
\text{edf}(\hat{\mu}) = \frac{1}{\sigma^2} \sum_{i=1}^{n} \text{cov}(\hat{\mu}_i,y_i).
\label{eq:edf}
\end{equation}
It is easy to verify that edf for the linear regression of $y$ upon a prespecified $X$ is the number of estimated coefficients $p$. However, \citet{janson2015effective} showed in simulations that there can be a large discrepancy between the edf of BS at size $k$ and $k$ itself. Similar evidence can be found in \citet{Tibshirani2015}, where the author quantifies the difference as the search degrees of freedom, which accommodates the amount of searching that BS performs in order to choose a best $k$-predictor subset. Unfortunately, the edf of BS does not have an analytical expression except when $X$ is orthogonal and $\mu=0$ \citep{Ye1998}. Numerically, we can apply tools like data perturbation \citep{Ye1998}, bootstrap \citep{Efron2004} or data randomization \citep{Harris2016} to estimate edf, but all rely on tunings of some hyperparameters and can be computationally intensive. 

This paper is motivated by the above challenges. We propose a novel least squares (LS)-based subset selection method, best orthogonalized subset selection (BOSS), that has the computational cost of a single ordinary LS fit. We also propose a heuristic degrees of freedom (hdf) for BOSS, which makes it feasible to use information criteria to select the optimal subset of BOSS. We show that a Kullback-Leibler (KL)-based information criterion, AICc-hdf, provides the best finite sample performance compared to other types of information criteria and CV, and we further demonstrate that BOSS using AICc-hdf is competitive compared to other types of LS-based methods and regularization methods. 

\subsection{Optimism theorem and information criteria for BS}
\label{sec:optimism}
Information criteria are designed to provide an unbiased estimate of the prediction (or testing) error, and can be derived from the so-called optimism theorem. Denote $\Theta$ as an error measure, $\text{err}$ as the training error, $\text{Err}$ as the testing error, $y^0$ as a new response vector with the same distribution but independent of the original $y$, and $E_0$ is the expectation taken over $y^0$. \citet{Efron1986} defined the optimism as 
\begin{equation*}
\text{op} = \text{Err} - \text{err},
\end{equation*}
and introduced the optimism theorem,
\begin{equation*}
E(\text{op}) = E(\text{Err}) - E(\text{err}).
\end{equation*}
A straightforward result from the optimism theorem is that 
\begin{equation}
\widehat{\text{Err}} = \text{err} + E(\text{op})
\label{eq:err_eop}
\end{equation}
is an unbiased estimator of $E(\text{Err})$, and is intended to balance the trade-off between model fit and model complexity. The challenge is to find $E(\text{op})$ for a given fitting rule $\hat{\mu}$ and error measure $\Theta(y,\hat{\mu})$.

When the error measure $\Theta$ is the squared error (SE), i.e. $\Theta(y_i,\hat{\mu}_i)=(y_i-\hat{\mu}_i)^2$, $\text{err}_{\text{SE}}$ (denoted as the training error when $\Theta$ is SE) then becomes the residual sum of squares $\text{RSS} = \sum_{i=1}^{n} \Theta(y_i, \hat{\mu}_i)$, and the testing error $\text{Err}_\text{SE} =\sum_{i=1}^n E_0[\Theta(y^0_i,\hat{\mu}_i)]$. \citet{Ye1998} and \citet{Efron2004} proved that for a general fitting rule $\hat{\mu}$ such as BS, $E(\text{op}_{\text{SE}})=2\sigma^2 \cdot \text{edf}(\hat{\mu})$, and hence $\widehat{\text{Err}}_{\text{SE}}$ in \eqref{eq:err_eop} becomes C$_p$-edf, where 
\begin{equation}
\text{C}_p\text{-edf} = \text{RSS} + 2 \sigma^2 \cdot \text{edf}.
\label{eq:cp_edf}
\end{equation}
These authors also showed that the traditional C$_p$,
\begin{equation*}
\text{C}_p\text{-ndf} = \text{RSS} + 2 \sigma^2 \cdot \text{ndf},
\end{equation*}
can be greatly biased when applied for BS, where ndf is the ``naive degrees of freedom'' that ignores any searching over model fits the procedure does in a given application. This is because C$_p$-ndf \citep{mallows1973some} was derived for a linear estimation rule $\hat{\mu} = Hy$ where $H$ is independent of $y$, which is not the case for BS. Here ndf equals $\text{Tr}(H)$. A further major issue regarding applying C$_p$ in practice is that it requires an estimate of $\sigma^2$. 

Another commonly used error measure is the deviance, which is (up to a constant)
\begin{equation}
\Theta = -2 \log f(y|\mu,\sigma^2),
\label{eq:deviance_def}
\end{equation}
\sloppy where $f$ is a pre-specified parametric model. Let $\hat{\mu}$ and $\hat{\sigma}^2$ be the maximum likelihood estimators obtained by maximizing $f(y|\mu,\sigma^2)$. We then have $\text{err}_{\text{KL}} = -2 \log f (y|\hat{\mu},\hat{\sigma}^2)$ and $\text{Err}_{\text{KL}}  = -2 E_0 \left[ \log f(y^0|\hat{\mu},\hat{\sigma}^2)\right] $, where the latter is the Kullback-Leibler (KL) discrepancy. For a linear estimation procedure, assuming asymptotic normality of $\hat{\mu}$ and $\hat{\sigma}^2$ ($f$ not necessarily Gaussian) and the true model distribution being contained in the specified parametric model $f$, \citet{konishi2008information} proved that $E(\text{op}_{\text{KL}}) = 2 \cdot \text{ndf} + o(1)$, and AIC \citep{Akaike1973},
\begin{equation*}
-2 \log f (y|\hat{\mu},\hat{\sigma}^2) + 2 \cdot \text{ndf},
\end{equation*}
asymptotically equals $\widehat{\text{Err}}_\text{KL}$ \eqref{eq:err_eop}. If $f$ follows a Gaussian distribution, as assumed in \eqref{eq:truemodel_def}, AIC can be expressed as
\begin{equation*}
\text{AIC-ndf} = n \log\left(\frac{\text{RSS}}{n}\right) + 2 \cdot \text{ndf}.
\end{equation*}
\citet{Hurvich1989} replaced the asymptotic $E(\text{op}_\text{KL})$ with its exact value, for Gaussian linear regression with an assumption that the predictors with non-zero true coefficients are included in the model, and proposed using the corrected AIC
\begin{equation*}
\text{AICc-ndf} = n \log\left(\frac{\text{RSS}}{n}\right) + n \frac{n+\text{ndf}}{n-\text{ndf}-2}.
\label{eq:aicc_ndf}
\end{equation*}
Neither AIC nor AICc has a penalty term depending upon $\sigma^2$, a clear advantage over C$_p$.

It remains a challenge to derive a KL-based information criterion for BS. \citet{Liao2018} estimated $E(\text{op}_{\text{KL}})$ via Monte Carlo simulations, but this relies on thousands of fits of the procedure, which is not computationally feasible for large datasets. 

In this work, we propose the use of AICc-edf
\begin{equation}
\text{AICc-edf} = n \log\left(\frac{\text{RSS}}{n}\right) + n \frac{n+\text{edf}}{n-\text{edf}-2}
\label{eq:aicc_edf}
\end{equation}
for this purpose. We demonstrate that $E(\text{AICc-edf})$ approximates $E(\text{Err}_\text{KL})$ well for BS. Moreover, both AICc-edf and $\widehat{\text{Err}}_\text{KL}$ generally choose the same subset when used as selection rules. Furthermore, a feasible implementation AICc-hdf (replacing edf with hdf in \eqref{eq:aicc_edf}) works reasonably well as a selection rule for BS with an orthogonal $X$ and for our proposed method BOSS with a general $X$.

\subsection{The structure and contributions of this paper}
The rest of the paper is organized as follows. In Section \ref{sec:ic_bs}, we introduce the hdf for BS in the presence of an orthogonal $X$, and we show in simulations that AICc-hdf provides the best finite sample performance compared to other types of information criteria and CV. In Section \ref{sec:boss}, we consider a general $X$ and propose the method BOSS. We provide numerical evidence that AICc-hdf is a reasonable selection rule for BOSS. Furthermore, we compare the performance of BOSS with that of LS-based methods and regularization methods in simulations and real data examples. In Section \ref{sec:justification_aicc_hdf}, 
we provide a theoretical justification for hdf in a restricted scenario, and numerical justifications in general situations. We also justify AICc-edf as an approximation of $E(\text{Err}_\text{KL})$. Lastly, we provide conclusions and potential future works in Section \ref{sec:conclusion}.

Below is guidance for applying LS-based methods in practice for data analysts.
\begin{itemize}
	\item Using information criteria in a naive way by plugging in the subset size as the degrees of freedom can lead to significantly worse performance than using edf and the feasible hdf.
	\item AICc is a better selection rule in terms of predictive performance in comparison to C$_p$, and the advantage is particularly strong when $p$ is large.
	\item AICc is not only more computationally efficient than cross-validation (CV), but also can result in subsets with better predictive performance, especially when the signal-to-noise ratio (SNR) is high or the sample size $n$ is large. The SNR is defined as $\text{Var}(x^T \beta) / \sigma^2$.
	\item BOSS using AICc is generally the best LS-based method in comparison to BS and forward stepwise (FS, in which the set of candidate models is chosen by stepping in variables one at a time based on choosing the one that adds the most fit to the regression at that step) using CV as the selection rule, in terms of both computational efficiency and predictive performance.
	\item Compared to regularization methods, BOSS using AICc performs the best when SNR is high or the sample size $n$ is large. In terms of support recovery in a sparse true model, BOSS recovers the true predictors and rarely includes any false positives. In contrast, regularization methods generally overfit. 
\end{itemize}


\section{Information criteria for BS}
\label{sec:ic_bs}

The best orthogonalized subset selection (BOSS) method first orthogonalizes an ordered set of predictors, and performs BS upon the orthogonalized predictors. Under an orthogonal $X$, BS essentially ranks the predictors based on the magnitude of their regression coefficients, i.e. the coefficient vector at subset size $k$ is $\hat{\beta}_i(k)=z_i \mathbbm{1}_{(|z_{i}| \ge |z_{(k)}|)}$, where $z=X^T y$ and $z_{(k)}$ is the $k$-th largest coefficient in absolute value. In this section, we study the use of information criteria to select the optimal subset size $k$ for BS under orthogonal predictors, and the results provide the foundations for the selection rule of BOSS. 

Using information criteria that are designed under a prefixed set of predictors to choose the number of predictors among the best models for each number of predictors (as is done in BS) results in inferior performance. \citet{Ye1998} showed that C$_p$-ndf is typically biased, and proposed using C$_p$-edf. We show in this section that the danger of using information criteria in the naive way also exists for other types of information criteria, such as AICc. Since edf for BS does not have an analytical expression, we propose a heuristic degrees of freedom (hdf) that can be calculated via an analytically-based formula. We show via simulations that the subset selected using AICc-hdf (by replacing edf with hdf in \eqref{eq:aicc_edf}) has the best predictive performance compared to other types of information criteria and CV. We present the main results in this section and give the justifications of hdf and AICc in Section \ref{sec:justification_aicc_hdf}.

\subsection{Heuristic degrees of freedom (hdf) for BS under orthogonal predictors}
\label{sec:hdf}

The edf of BS has an analytical expression only when $X$ is orthogonal and the true model is $\mu=0$ \citep{Ye1998}. \citet{Tibshirani2015} studied the Lagrangian formulation of BS (LBS) and provided an analytical expression for edf without any restrictions on $\mu$ (the only assumption being that $X$ is orthogonal). To distinguish between the two methods, we use df$_C(k)$ and df$_L(\lambda)$ to denote edf of BS for subset size $k$ and edf of LBS for tuning parameter $\lambda$, respectively. 

For each regularization parameter $\lambda \ge 0$, LBS solves
\begin{equation}
	\min_\beta \frac{1}{2} \lVert \mathbf{y}-\mathbf{X\beta}\rVert_2^2 + \lambda\lVert \mathbf{\beta} \rVert_0.
\label{eq:lbestsubset-setup}
\end{equation} 
Both LBS \eqref{eq:lbestsubset-setup} and BS \eqref{eq:bestsubset-setup} are LS regressions of $y$ upon a certain subset of $X$. With orthogonal $X$, both problems have analytical solutions: $\hat{\beta}_i(\lambda)=z_i \mathbbm{1}_{(|z_{i}| \ge \sqrt{2\lambda})}$ for \eqref{eq:lbestsubset-setup} and $\hat{\beta}_i(k)=z_i \mathbbm{1}_{(|z_{i}| \ge |z_{(k)}|)}$ for \eqref{eq:bestsubset-setup}, where $z=X^T y$ and $z_{(k)}$ is the $k$-th largest coefficient in absolute value. These two problems are not equivalent, and there is no clear one-to-one correspondence between $\lambda$ in \eqref{eq:lbestsubset-setup} and $k$ in \eqref{eq:bestsubset-setup}. Indeed, for each $\lambda$ there exists a $k$ such that $\hat{\beta}(\lambda) = \hat{\beta}(k)$ where $\hat{\beta}(\lambda)$ is the solution of \eqref{eq:lbestsubset-setup} at $\lambda$ and $\hat{\beta}(k)$ is the solution of \eqref{eq:bestsubset-setup} at $k$, but the reverse does not necessarily hold, since there will be multiple $\lambda$ corresponding to the same solution $\hat{\beta}(k)$. Moreover, with a general $X$, solving \eqref{eq:lbestsubset-setup} does not guarantee recovery of the entire solution path given by solving \eqref{eq:bestsubset-setup} for $k=0,\dots,K$.

Under an orthogonal $X$, \citet{Tibshirani2015} derived an expression for df$_L(\lambda)$ as 
\begin{equation}
	\text{df}_L(\lambda) = E(k_L(\lambda)) + \frac{\sqrt{2\lambda}}{\sigma} \sum_{i=1}^{K} \left[\phi\left(\frac{\sqrt{2\lambda}-(X^T \mu)_i}{\sigma}\right) + \phi\left(\frac{-\sqrt{2\lambda}-(X^T \mu)_i}{\sigma}\right) \right],
	\label{eq:thdf_expression}
\end{equation}
where the expected subset size in the expression is given by 
\begin{equation}
	E(k_L(\lambda)) = \sum_{i=1}^{K} \left[1-\Phi\left(\frac{\sqrt{2\lambda}-(X^T \mu)_i}{\sigma}\right) + \Phi\left(\frac{-\sqrt{2\lambda}-(X^T \mu)_i}{\sigma}\right) \right].
	\label{eq:thdf_size_expression}
\end{equation}
Given the similarity of problems \eqref{eq:bestsubset-setup} and \eqref{eq:lbestsubset-setup}, we would like to approximate df$_C(k)$ with df$_L(\lambda)$. One implementation of this proceeds as follows. Note that df$_C(k)$ is a discrete function of $k=0,\cdots,K$ while df$_L(\lambda)$ is a continuous function of a real variable $\lambda\ge 0$. We propose an hdf that uses $\text{df}_L(\lambda)$ for a particular value of $\lambda$ depending on $k$ as a proxy for $\text{df}_C(k)$. Based on \eqref{eq:thdf_size_expression}, $\lambda$ and $E(k_L(\lambda))$ have a clear one-to-one correspondence, which implies that we can find a unique $\lambda_k^\star$ such that $E(k_L(\lambda_k^\star)) = k$ for each $k=1,\cdots,K$. The value of hdf is df$_L(\lambda_k^\star)$ obtained by substituting $\lambda^\star_k$ into \eqref{eq:thdf_expression}. We also let hdf$(0)=0$ since df$_C(0)=0$. The implementation process is summarized in Algorithm \ref{alg:hdf}. In place of $\mu$ and $\sigma$, we use the OLS estimates based on the full model, i.e. $\hat{\mu}=XX^T y$, $\hat{\sigma}^2 = \lVert y-\hat{\mu} \rVert_2^2/(n-p)$.

\begin{algorithm}
	\caption{The heuristic df (hdf) for BS under orthogonal predictors}\label{alg:hdf}
	Input: $X$ (orthogonal), $\sigma$ and $\mu$. For a given subset size $k$, 
	\begin{enumerate}[label=\arabic*.]
		\item Based on \eqref{eq:thdf_size_expression}, calculate $\lambda_k^\star$ such that $E(k_L(\lambda_k^\star)) = k$.
		\item Based on \eqref{eq:thdf_expression}, calculate hdf$(k) = \text{df}_L(\lambda_k^\star)$.
	\end{enumerate}
	Repeat the above steps for $k=1,\cdots,K$ and let hdf$(0)=0$, yielding hdf for each subset. 
	
\end{algorithm}

\subsection{Information criteria for BS under orthogonal predictors}
\label{sec:aicc_performance_bs}

The hdf makes it feasible to use information criteria to select the optimal subset for BS. For instance, by replacing edf with hdf in the expression of C$_p$-edf and AICc-edf, we have the feasible criteria C$_p$-hdf and AICc-hdf, respectively. In this section, we perform a comprehensive set of simulation studies to compare the performance of different information criteria for BS, and show that AICc-hdf provides the best finite sample performance. We include infeasible versions of the criteria based on the edf (these are infeasible since the edf would not be known in real data applications), since those would correspond to the ideal versions of the criteria. To calculate the edf, we fit the BS procedure on $1000$ replicated realizations of the response generated from the true model after fixing $X$, and estimate the definition \eqref{eq:edf} using the sample covariance. We also consider a numerical estimation of edf that is based on the parametric bootstrap, and we denote it as bdf. The detailed implementation of bdf and the benefit of parametric bootstrap are discussed in \citet{Efron2004}. In our experiment, we use $100$ bootstrapped samples. Also, by analogy to C$_p$ and AICc, we define BIC-edf as
\begin{equation*}
\text{BIC-edf} = n \log\left(\frac{\text{RSS}}{n}\right) + \log(n) \cdot \text{edf},
\end{equation*}
where the original BIC (or BIC-ndf in our notation) was introduced in \citet{schwarz1978estimating}. In addition to the information criteria, we also include 10-fold CV for comparison. Note that the CV results are only available for $p \le 30$, since orthogonality no longer holds for random subsamples and BS is therefore fitted using the ``leaps'' algorithm.

We consider two sparse true models (denoted as Orth-Sparse-Ex1 and Orth-Sparse-Ex2) that have $p_0=6$ signal predictors (those with non-zero coefficients), and a dense true model (denoted as Orth-Dense) where all predictors have non-zero coefficients. We also consider three signal-to-noise (SNR) ratios, and the SNR is defined as $\text{Var}(x^T \beta)/\sigma^2$. The average oracle $R^2$ (linear regression on the set of true predictors) corresponding to these three SNR values are roughly $20\%$, $50\%$ and $90\%$. We further consider eight combinations of $(n,p)$, resulting in $72$ different scenarios in this experiment. In each scenario, $1000$ replications of the response $y$ are generated by fixing the design matrix $X$. A fitting procedure $\hat{\mu}$ is evaluated via the average RMSE, where 
\begin{equation*}
\text{RMSE}(\hat{\mu}) = \sqrt{ \frac{1}{n} \lVert \hat{\mu}-X\beta \rVert_2^2}.
\end{equation*}
We summarize the results using two relative measures. The $\%$ worse than best possible BS defines the relative performance of the procedure relative to the BS model, where on a single fit, the subset with the minimum RMSE among all candidates is selected, as if an oracle tells us the best model. The relative efficiency defines the performance of the procedure relative to all other procedures considered in the experiment. It is a measure between $0$ and $1$, and higher value indicates better performance. We also present the sparsistency (number of true positives) and number of extra predictors (number of false positives). The details of the simulation setup and evaluation metrics are discussed in the Supplemental Material Section \ref{sec:simulation_setup_orthx}, and the complete simulation results are presented in the Online Supplemental Material\footnote{The complete results for simulation studies in this paper are available at \url{https://github.com/sentian/BOSSreg}.}. 

\begin{table}[ht]
\centering
\caption{Information criteria for BS under orthogonal $X$. The true model setup is Orth-Sparse-Ex1 (see Supplemental Material Section \ref{sec:simulation_setup_orthx} for details). 
                    The columns involving ``edf'' refer to infeasible selection rules since edf is estimated as if the true model is known, 
                    while other columns correspond to feasible rules.} 
\label{tab:ic_df_orthx_sparseex1}
\scalebox{0.7}{
\begin{tabular}{|c|c|c|cc|cc|cc|c|}
  \toprule 
 \multicolumn{1}{|r}{} & \multicolumn{1}{c}{} &       & \multicolumn{2}{c|}{C$_p$} & \multicolumn{2}{c|}{AICc} & \multicolumn{2}{c|}{BIC} & \multirow{2}[4]{*}{CV} \\
 \cmidrule{4-9}\multicolumn{1}{|r}{} & \multicolumn{1}{c}{} &       & edf   & ndf/hdf/bdf & edf   & ndf/hdf/bdf & edf   & ndf/hdf/bdf &     \\
 \midrule 
 \multicolumn{1}{|r}{} & \multicolumn{1}{c}{} &       & \multicolumn{7}{c|}{\% worse than the best possible BS} \\
 \midrule 
 \multirow{4}[4]{*}{n=200} & \multirow{2}[2]{*}{hsnr} & p=30 & 4 & 84/5/7 & 2 & 83/2/5 & 0 & 28/0/0 & 24 \\ 
   &  & p=180 & 1 & 338/30/32 & 0 & 392/1/2 & 0 & 206/0/0 & - \\ 
  \cmidrule{2-10} & \multirow{2}[2]{*}{lsnr} & p=30 & 20 & 25/36/33 & 21 & 24/37/35 & 68 & 23/68/67 & 28 \\ 
   &  & p=180 & 15 & 108/35/34 & 18 & 132/22/22 & 25 & 50/25/25 & - \\ 
  \midrule \multirow{4}[4]{*}{n=2000} & \multirow{2}[2]{*}{hsnr} & p=30 & 3 & 85/3/6 & 3 & 85/3/6 & 0 & 9/0/0 & 23 \\ 
   &  & p=180 & 0 & 334/1/3 & 1 & 337/1/3 & 0 & 60/0/0 & - \\ 
  \cmidrule{2-10} & \multirow{2}[2]{*}{lsnr} & p=30 & 3 & 85/6/7 & 3 & 85/5/6 & 0 & 9/0/0 & 23 \\ 
   &  & p=180 & 0 & 334/5/4 & 1 & 337/4/4 & 0 & 60/1/1 & - \\ 
   \midrule 
 \multicolumn{1}{|r}{} & \multicolumn{1}{r}{} &       & \multicolumn{7}{c|}{Relative efficiency} \\
 \midrule 
\multirow{4}[4]{*}{n=200} & \multirow{2}[2]{*}{hsnr} & p=30 & 0.96 & 0.54/0.95/0.93 & 0.98 & 0.55/0.98/0.96 & 1 & 0.78/1/1 & 0.81 \\ 
   &  & p=180 & 0.99 & 0.23/0.77/0.76 & 1 & 0.2/0.99/0.98 & 1 & 0.33/1/1 & - \\ 
  \cmidrule{2-10} & \multirow{2}[2]{*}{lsnr} & p=30 & 1 & 0.97/0.89/0.9 & 1 & 0.97/0.88/0.89 & 0.72 & 0.98/0.72/0.72 & 0.94 \\ 
   &  & p=180 & 1 & 0.55/0.86/0.86 & 0.97 & 0.5/0.95/0.95 & 0.93 & 0.77/0.93/0.93 & - \\ 
  \midrule \multirow{4}[4]{*}{n=2000} & \multirow{2}[2]{*}{hsnr} & p=30 & 0.97 & 0.54/0.97/0.94 & 0.97 & 0.54/0.97/0.94 & 1 & 0.92/1/1 & 0.81 \\ 
   &  & p=180 & 1 & 0.23/0.99/0.97 & 0.99 & 0.23/0.99/0.97 & 1 & 0.62/1/1 & - \\ 
  \cmidrule{2-10} & \multirow{2}[2]{*}{lsnr} & p=30 & 0.97 & 0.54/0.95/0.94 & 0.97 & 0.54/0.95/0.94 & 1 & 0.92/1/1 & 0.81 \\ 
   &  & p=180 & 1 & 0.23/0.96/0.96 & 0.99 & 0.23/0.96/0.96 & 1 & 0.62/0.99/0.99 & - \\ 
   \midrule 
 \multicolumn{1}{|r}{} & \multicolumn{1}{r}{} &       & \multicolumn{7}{c|}{Sparsistency (number of extra variables)} \\
 \midrule 
\multirow{4}[4]{*}{n=200} & \multirow{2}[2]{*}{hsnr} & p=30 & 6(0.1) & 6(3.9)/6(0.2)/6(0.2) & 6(0) & 6(3.8)/6(0.1)/6(0.1) & 6(0) & 6(0.6)/6(0)/6(0) & 6(0.7) \\ 
   &  & p=180 & 6(0) & 6(32.2)/6(6.4)/6(6.3) & 6(0) & 6(48.9)/6(0)/6(0) & 6(0) & 6(9.5)/6(0)/6(0) & - \\ 
  \cmidrule{2-10} & \multirow{2}[2]{*}{lsnr} & p=30 & 4.5(1.9) & 5.3(3.9)/4.2(4.9)/4.2(4) & 4.2(1.2) & 5.2(3.8)/3.3(2.2)/3.4(1.8) & 0.1(0) & 3.7(0.6)/0.1(0)/0.2(0) & 4(1.9) \\ 
   &  & p=180 & 1.9(0.5) & 5.3(32.2)/1.8(10.9)/1.9(9.8) & 1.1(0.1) & 5.6(49)/0.5(0)/0.6(0) & 0(0) & 4.2(8.4)/0(0)/0(0) & - \\ 
  \midrule \multirow{4}[4]{*}{n=2000} & \multirow{2}[2]{*}{hsnr} & p=30 & 6(0.1) & 6(3.8)/6(0.1)/6(0.2) & 6(0.1) & 6(3.8)/6(0.1)/6(0.2) & 6(0) & 6(0.1)/6(0)/6(0) & 6(0.6) \\ 
   &  & p=180 & 6(0) & 6(27.5)/6(0)/6(0) & 6(0) & 6(28.2)/6(0)/6(0) & 6(0) & 6(1.1)/6(0)/6(0) & - \\ 
  \cmidrule{2-10} & \multirow{2}[2]{*}{lsnr} & p=30 & 6(0.1) & 6(3.8)/6(0.2)/6(0.2) & 6(0.1) & 6(3.8)/6(0.2)/6(0.2) & 6(0) & 6(0.1)/6(0)/6(0) & 6(0.6) \\ 
   &  & p=180 & 6(0) & 6(27.5)/6(0.1)/6(0) & 6(0) & 6(28.2)/6(0.1)/6(0) & 6(0) & 6(1.1)/6(0)/6(0) & - \\ 
   \bottomrule 
\end{tabular}
}
\end{table}

\begin{table}[ht]
\centering
\caption{Information criteria for BS under orthogonal $X$. The true model setup is Orth-Dense (see Supplemental Material Section \ref{sec:simulation_setup_orthx} for details).} 
\label{tab:ic_df_orthx_dense}
\scalebox{0.7}{
\begin{tabular}{|c|c|c|cc|cc|cc|c|}
  \toprule 
 \multicolumn{1}{|r}{} & \multicolumn{1}{c}{} &       & \multicolumn{2}{c|}{C$_p$} & \multicolumn{2}{c|}{AICc} & \multicolumn{2}{c|}{BIC} & \multirow{2}[4]{*}{CV} \\
 \cmidrule{4-9}\multicolumn{1}{|r}{} & \multicolumn{1}{c}{} &       & edf   & ndf/hdf/bdf & edf   & ndf/hdf/bdf & edf   & ndf/hdf/bdf &     \\
 \midrule 
 \multicolumn{1}{|r}{} & \multicolumn{1}{c}{} &       & \multicolumn{7}{c|}{\% worse than the best possible BS} \\
 \midrule 
 \multirow{4}[4]{*}{n=200} & \multirow{2}[2]{*}{hsnr} & p=30 & 1 & 11/1/2 & 1 & 13/1/2 & 1 & 28/3/5 & 7 \\ 
   &  & p=180 & 7 & 45/21/20 & 9 & 52/18/19 & 18 & 26/39/42 & - \\ 
  \cmidrule{2-10} & \multirow{2}[2]{*}{lsnr} & p=30 & 15 & 10/16/16 & 20 & 10/21/20 & 27 & 16/27/27 & 16 \\ 
   &  & p=180 & 8 & 86/22/22 & 7 & 102/7/7 & 7 & 39/7/7 & - \\ 
  \midrule \multirow{4}[4]{*}{n=2000} & \multirow{2}[2]{*}{hsnr} & p=30 & 0 & 1/0/0 & 0 & 1/0/0 & 0 & 18/0/1 & 1 \\ 
   &  & p=180 & 6 & 34/8/8 & 6 & 34/8/8 & 19 & 7/36/37 & - \\ 
  \cmidrule{2-10} & \multirow{2}[2]{*}{lsnr} & p=30 & 2 & 11/3/3 & 2 & 11/3/3 & 44 & 41/36/45 & 10 \\ 
   &  & p=180 & 8 & 48/10/10 & 8 & 48/10/10 & 24 & 8/45/47 & - \\ 
   \midrule 
 \multicolumn{1}{|r}{} & \multicolumn{1}{r}{} &       & \multicolumn{7}{c|}{Relative efficiency} \\
 \midrule 
\multirow{4}[4]{*}{n=200} & \multirow{2}[2]{*}{hsnr} & p=30 & 1 & 0.91/1/1 & 1 & 0.9/1/0.99 & 1 & 0.79/0.98/0.96 & 0.95 \\ 
   &  & p=180 & 1 & 0.74/0.89/0.89 & 0.99 & 0.71/0.91/0.9 & 0.91 & 0.85/0.77/0.76 & - \\ 
  \cmidrule{2-10} & \multirow{2}[2]{*}{lsnr} & p=30 & 0.95 & 1/0.95/0.95 & 0.91 & 1/0.91/0.91 & 0.86 & 0.94/0.86/0.86 & 0.94 \\ 
   &  & p=180 & 1 & 0.58/0.88/0.88 & 1 & 0.53/1/1 & 1 & 0.77/1/1 & - \\ 
  \midrule \multirow{4}[4]{*}{n=2000} & \multirow{2}[2]{*}{hsnr} & p=30 & 1 & 0.99/1/1 & 1 & 0.99/1/1 & 1 & 0.85/1/0.99 & 0.99 \\ 
   &  & p=180 & 1 & 0.79/0.98/0.98 & 1 & 0.79/0.98/0.98 & 0.89 & 1/0.78/0.78 & - \\ 
  \cmidrule{2-10} & \multirow{2}[2]{*}{lsnr} & p=30 & 1 & 0.92/0.99/0.99 & 1 & 0.92/0.99/0.99 & 0.71 & 0.73/0.75/0.7 & 0.93 \\ 
   &  & p=180 & 1 & 0.73/0.98/0.98 & 1 & 0.73/0.98/0.98 & 0.87 & 1/0.74/0.73 & - \\ 
   \midrule 
 \multicolumn{1}{|r}{} & \multicolumn{1}{r}{} &       & \multicolumn{7}{c|}{Sparsistency (number of extra variables)} \\
 \midrule 
\multirow{4}[4]{*}{n=200} & \multirow{2}[2]{*}{hsnr} & p=30 & 30 & 24.7/29.5/29 & 30 & 24.2/29.4/28.8 & 30 & 20.9/28.8/27.5 & 26.6 \\ 
   &  & p=180 & 20.5 & 53.3/37.4/35.5 & 18.3 & 62.3/16.3/16.3 & 16.1 & 35/13.7/13.5 & - \\ 
  \cmidrule{2-10} & \multirow{2}[2]{*}{lsnr} & p=30 & 12.8 & 10.5/14.6/13 & 7.6 & 10.3/8.5/7.6 & 0 & 4/0/0 & 7.5 \\ 
   &  & p=180 & 0.8 & 39/14.5/13.7 & 0.3 & 55.2/0.2/0.3 & 0 & 11.8/0/0 & - \\ 
  \midrule \multirow{4}[4]{*}{n=2000} & \multirow{2}[2]{*}{hsnr} & p=30 & 30 & 29.8/30/29.9 & 30 & 29.8/30/29.9 & 30 & 28.6/30/29.9 & 29.8 \\ 
   &  & p=180 & 32.1 & 58.9/32.4/32.3 & 31.8 & 58.9/31.6/31.6 & 27 & 31.3/25/24.9 & - \\ 
  \cmidrule{2-10} & \multirow{2}[2]{*}{lsnr} & p=30 & 28.8 & 19.9/28.2/26.9 & 28.8 & 19.9/28.1/26.8 & 13.5 & 12.5/16.7/14.1 & 22.3 \\ 
   &  & p=180 & 13.9 & 43.8/14/14 & 13.6 & 44.1/13.3/13.3 & 9.1 & 13.4/7/6.8 & - \\ 
   \bottomrule 
\end{tabular}
}
\end{table}

A selected set of results is given in Tables \ref{tab:ic_df_orthx_sparseex1} and \ref{tab:ic_df_orthx_dense}. Here ``hsnr'' and ``lsnr'' represent the high and low SNR, respectively. A brief summary of the results is as follows:

\begin{itemize}
	\item Using information criteria in the naive way (with ndf) can be dangerous, especially when $p$ is large and SNR is high. For example, using ndf in AICc significantly overfits and can be almost $400$ times worse in terms of RMSE than using hdf for $n=200$, high SNR and $p=180$ in the sparse example. Increasing the sample size $n$ does not improve the performance of naive implementation of information criteria, and the overfiting persists.
	\item AICc-hdf generally does not lose much efficiency and performs similarly in terms of RMSE, in comparison to the infeasible AICc-edf. Increasing the sample size $n$ or SNR improves the performance of both AICc-edf and AICc-hdf. 
	\item AICc-hdf performs very similarly to AICc-bdf. Since bdf is calculated based on $100$ bootstrapped samples, it is roughly $100$ times more intensive than hdf in computations. 
	\item AICc-hdf is generally better than 10-fold CV, e.g. when $n$ is large or SNR is high. Note that 10-fold CV is roughly $10$ times heavier in terms of computation than AICc-hdf. It is also worth noticing that these findings are broadly consistent with the results reported by \citet{Taddy2017} for the gamma lasso method. 
	\item C$_p$-edf performs similarly to AICc-edf. In contrast, when we consider the feasible implementations (ndf/hdf/bdf), i.e. when $\sigma$ is estimated by full OLS, C$_p$ can suffer when $p$ is close to $n$, such as when $n=200$ and $p=180$. 
	\item Under a sparse true model BIC-hdf performs slightly better than AICc-hdf except when SNR is low and $n=200$, where BIC is considerably worse. Under a dense true model BIC-hdf is always outperformed by AICc-hdf. 
\end{itemize}
For the reasons presented above, we conclude that AICc-hdf is the best feasible selection rule for BS on orthogonal predictors, among all that have been considered. 


\section{Best orthogonalized subset selection (BOSS)}
\label{sec:boss}

For predictors in general position, BS is not computationally feasible for large problems. In this section, we propose a feasible LS-based subset selection method BOSS that is based on orthogonalizing the predictors. The orthogonalization not only makes the computation of BS feasible, but also allows us to take advantage of the superior performance of AICc-hdf as discussed in Section \ref{sec:aicc_performance_bs}. BOSS using AICc-hdf as the selection rule has computational cost of the order $O(npK)$, that is of the same order as a multiple regression on all $p$ predictors for the problem where $n>p$, and on a selected subset of $n$ predictors for $n \le p$. We further demonstrate the competitive performance of BOSS via simulations and real world examples. 

The main steps for BOSS can be summarized as follows: 1) order and orthogonalize the predictors, 2) perform BS on the set of orthogonalized predictors, 3) transform the coefficients back to the original space, and 4) use a selection rule such as AICc-hdf to choose the optimal single subset. 

\subsection{The solution path of BOSS}
\label{sec:boss_solutionpath}

BOSS starts by ordering and orthogonalizing the predictors, taking $K$ steps in total. The ordering is based on partial correlations with the response $y$, and the orthogonalization is based on QR decomposition with Gram-Schmidt. For step $k$, we use $X_{S_{k-1}}$ to denote the set of ordered predictors from the previous step, and use $Q_{S_{k-1}}$ to denote the orthogonal basis of $X_{S_{k-1}}$. From the remaining $p-k+1$ predictors, we choose the one that has the largest correlation with $y$ conditioning on $Q_{S_{k-1}}$, that is the correlation between $y$ and the residual from regressing a candidate predictor on $Q_{S_{k-1}}$. This costs $O(n)$ since we maintain the regression result, e.g. estimated coefficients and residual, in the previous steps. Repeating the above step for all $p-k+1$ predictors costs $O(n(p-k+1))$. We then update the QR decomposition, by adding the chosen predictor as a new column, which costs $O(n(K-k))$ via the modified Gram-Schmidt algorithm as discussed in \citet{hammarling2008updating}. After $K$ steps, we end up with an ordered set of predictors $X_{S_K}$ and its orthogonal basis $Q_{S_K}$, and the total cost for ordering and orthogonalization is $O(npK)$. We denote the regression coefficient vector of $y$ on $Q_{S_K}$ as $z$. BOSS then performs BS on $Q_{S_K}$, which is a ranking of predictors based on their absolute values of corresponding element in $z$, and the cost is $O(K\log(K))$. We use $\tilde{\gamma}(k_Q)$ to denote the BS coefficient vector at size $k_Q$, where $k_Q=1,\cdots,K$ specifies the subset size in the orthogonal space. Finally, BOSS transforms the coefficient vectors $\tilde{\gamma}=[\tilde{\gamma}(0),\cdots,\tilde{\gamma}(K)]$ back to the original space. Therefore, the total cost for the entire solution path of BOSS is on the order of $O(npK)$. The detailed implementation for obtaining the solution path is summarized as steps 1-5 in Algorithm \ref{alg:boss}.

\begin{algorithm}
	\caption{Best Orthogonalized Subset Selection (BOSS)}\label{alg:boss}
	\begin{enumerate}[label=\arabic*.]
		\item Standardize $y$ and the columns of $X$ to have mean $0$, and denote the means as $\bar{X}$ and $\bar{y}$.

		\textbf{Order and orthogonalize the predictors:}

		\item From the $p$ predictors, select the one that has the largest marginal correlation with the response $y$,  and denote it as $X_{S_1}$. Standardize $X_{S_1}$ to have unit $l_2$ norm and denote it as $Q_{S_1}$, where $S_1$ is the variable number of this chosen predictor. Calculate $R_{S_1}$ such that $X_{S_1} = Q_{S_1} R_{S_1}$. Let $S=\{1,\cdots, p\}$. Initialize vectors $\text{resid}_j=X_j$ where $j=1,\cdots,p$.
		\item For $k=2,\cdots,K$ ($K=\min\{n,p\}$):
		\begin{enumerate}[label=\alph*.]
			\item For each of the $p-k+1$ predictors $X_j$ in $X_{S \setminus S_{k-1} }$, calculate its partial correlations with the response $y$ conditioning on $Q_{S_{k-1}}$. 
			\begin{enumerate}[label=a\arabic*.]
				\item Regress $X_j$ on $Q_{S_{k-1} \setminus S_{k-2}}$ ($S_{k-2}=\emptyset$ if $k=2$), and denote the estimated coefficient as $r$. Update $\text{resid}_j = \text{resid}_j - r Q_{S_{k-1} \setminus S_{k-2}}$.
				\item Calculate the correlation between y and $\text{resid}_j$.
			\end{enumerate}
			\item Select the predictor that has the largest partial correlation in magnitude, augment $S_{k-1}$ with this predictor number and call it $S_{k}$.
			\item Update $Q_{S_{k-1}}$ and $R_{S_{k-1}}$ given the newly added column $X_{S_k \setminus S_{k-1}}$, and call them $Q_{S_k}$ and $R_{S_k}$. The update is based on the modified Gram-Schmidt algorithm as discussed in \citet{hammarling2008updating}.
		\end{enumerate}
		
		\textbf{BS on the orthogonalized predictors $Q_{S_{K}}$:}

		\item Calculate $\tilde{\gamma}_j (k_Q)= z_j \mathbbm{1}(|z_j| \ge |z_{(k_Q)}|)$, i.e. the $j$-th component of coefficient vector for subset size $k_Q$, where $z=Q_{S_{K}}^T y$ and $z_{(k_Q)}$ is the $k$-th largest entry in absolute values. Let $\tilde{\gamma} = [\tilde{\gamma} (0) \tilde{\gamma} (1) \cdots \tilde{\gamma} (K)]$.
		
		\textbf{Transform back to the original space:}

		\item Project $\tilde{\gamma}$, a $p \times (K+1)$ matrix, to the original space of $X_{S_{K}}$, i.e. back solving $R \tilde{B} = \tilde{\gamma}$, and re-order the rows of $\tilde{B}$ to their correspondences in $X$, i.e. $\hat{B} = O \tilde{B}$ where $O$ represents the ordering matrix s.t. $X_{S_{K}}=XO$. The intercept vector is $\hat{B}_0 = \bar{y} \mathbbm{1} - \hat{B}^T \bar{X}$. 

		\textbf{Select the subset:}

		\item Select the subset using AICc-hdf (replacing edf with hdf in \eqref{eq:aicc_edf}), where hdf is calculated via Algorithm \ref{alg:hdf}, by inputting $(Q_{S_{K}},\hat\sigma,\hat\mu)$. For $n > p$, we use OLS estimates based on the full model, i.e. $\hat{\mu}=Q_{S_{K}} z$, $\hat{\sigma}^2 = \lVert y-\hat{\mu} \rVert_2^2/(n-p)$. For $n \le p$, we use $\hat{\mu}=\hat{\mu}_l$ and $\hat{\sigma} = \sqrt{\lVert y-\hat{\mu}_l \rVert_2^2 / (n-\text{df}(\hat{\mu}_l)-1)}$ as discussed in \citet{reid2016study}, where $\hat{\mu}_l$ are the lasso fitted values based on 10-fold CV and df$(\hat{\mu}_l)$ is the corresponding number of non-zero coefficients in the lasso estimate. Note that the inclusion of an intercept term implies that hdf is increased by $1$. 
	\end{enumerate}
\end{algorithm}

The ordering of predictors is an important part of implementation of the algorithm. Consider a sparse true model with only two uncorrelated predictors $X=[X_1,X_2]$, $\beta=[0,1]^T$ and a high SNR. The best model in such a scenario is LS regression on $X_2$. Without the ordering step, the orthogonal basis is $Q=[Q_1,Q_2]$ s.t. $X=QR$, i.e. the predictors are orthogonalized in their physical orders. The one-predictor model ($k_Q=1$) of BS can either be $Q_1$ or $Q_2$, which when transformed back to the space of $X$ do not correspond to LS regression upon $X_2$. The former corresponds to LS estimates upon $X_1$, while the latter is a linear combination of LS estimates upon $X$ and LS estimates upon $X_1$; the former leads to a completely wrong model while the latter results in non-zero coefficients on both predictors. In contrast, if $X_2$ is the first variable orthogonalized, the best subset will be based on that variable alone, the correct choice. Therefore, the ordering step is crucial to both sparsity as well as predictive performance. It is worth noting that BOSS is flexible in terms of the choice of ordering rules. For instance, a less aggressive rule based on the lars algorithm \citep{efron2004least} could be adopted instead. We leave discussion of the benefits and drawbacks of different ordering rules to future work. Note that we show that the coefficients of BOSS can be expressed as a linear combination of the LS coefficients on subsets of $X$ in Theorem \ref{thm:correspondence}, the proof of which can be found in the Supplemental Material Section \ref{sec:correspondence}.

\begin{theorem}[Coefficients of BOSS are a linear combination of LS coefficients on subsets of $X$] 
	Suppose $X$ has full column rank and the columns are already ordered, i.e. $X=X_{S_{K}}$. $X=QR$ where $Q$ is an $n \times K$ matrix with orthonormal columns and $R$ is a $K \times K$ upper-triangular matrix. Let $S_k=\{j_1,j_2,\cdots,j_{k_Q}\}$ denote the support (position of predictors) of the best $k_Q$-predictor model given by BS upon $(Q,y)$, and use $\hat{\gamma}(k_Q)$ ($K$ by $1$) to denote the BS coefficients. The corresponding coefficients in the $X$ space, i.e. $\hat{\beta}(k_Q)$ s.t. $R\hat{\beta}(k_Q)=\hat{\gamma}(k_Q)$, can be expressed as
	\begin{equation*}
	\hat{\beta}(k_Q) = \sum_{j\in S_k} \left(\hat{\alpha}^{(j)} - \hat{\alpha}^{(j-1)}\right),
	\end{equation*}
	where the first $j$ entries in $\hat{\alpha}^{(j)}$ ($K$ by $1$) are LS coefficients of regressing $y$ upon $[X_1,X_2,\cdots,X_j]$ (the first $j$ columns in $X$), and the remaining $K-j$ entries are zero.

	\label{thm:correspondence}
\end{theorem}

\subsection{Connection to FS and the advantage of BOSS}

BOSS is closely related to FS. In fact, instead of performing BS on the set of orthogonalized predictors $Q_{S_{K}}$ in the fourth step of Algorithm \ref{alg:boss}, if we fit LS on the subset of $Q_{S_{K}}$ in a nested fashion, i.e. $\tilde{\gamma}_j (k_Q)= z_j \mathbbm{1}(j \le k_Q)$ and $z=Q_{S_{K}}^T y$, steps 1-5 of the algorithm provide the solution path of FS, and is similar to the orthogonal greedy algorithm discussed in \citet{ing2011stepwise}. Since BS on orthogonal predictors $Q_{S_K}$ is essentially ranking the predictors based on their LS coefficients ($O(K\log(K))$ operations), BOSS involves little additional computational cost compared to FS.

BOSS can, however, provide a better solution path than FS. At a given step, once a predictor is selected, it remains in the subsets of every following step of FS. In many circumstances, the greedy characteristic can lead to overfit, since noise predictors (those with $\beta_j=0$) step in during early steps. However, BS on the set of orthogonalized predictors gives the chance for BOSS to ``look back'' at the predictors that are already stepped in. By revisiting these predictors and allowing them to be dropped, BOSS can provide a solution path that is sparser, with better predictive performance compared to FS. 

We consider two numerical examples Sparse-Ex3 and Sparse-Ex4, where the true models are sparse. The true coefficient vectors for Sparse-Ex3 and Sparse-Ex4 are $\beta=[1_6,0_{p-6}]^T$ and $\beta=[1,-1,5,-5,10,-10,0_{p-6}]^T$, respectively. We consider a high SNR and a high correlation between predictors ($\rho=0.9$). For Sparse-Ex3, the signal predictors (those with non-zero coefficients) are pairwise correlated with noise predictors (correlation coefficient is denoted as $\rho$), while for Sparse-Ex4, the signal predictors are pairwise correlated with opposite effects. We generate the design matrix $X$ once, and draw $1000$ replications of the response $y$ based on \eqref{eq:truemodel_def}. The details of the model setup are given in Supplemental Material Section \ref{sec:simulation_setup_generalx}.

Figure \ref{fig:lossratio_fs_boss_k} shows the average RMSE along the solution paths of BS, FS and BOSS, for the two examples. When the true model is Sparse-Ex3, all three methods provide almost the same solution path. However, for Sparse-Ex4, we see a clear advantage of BOSS over FS in early steps up until about the fifteenth step. Recall that in Sparse-Ex4, there are $p_0=6$ predictors with $\beta_j \ne 0$ that are pairwise correlated with opposite effects, where each pair together leads to a high $R^2$ but each single one of them contributes little. When the correlation between the variables is high, the effect of one almost completely cancels out the effect of the other on $y$. Therefore all of the predictors (both true and noise predictors) have approximately zero marginal correlation with $y$, and they have equal chance of stepping in. Since the subsets along the solution path of FS are nested, if a noise predictor steps in during early steps, it remains in the subsets of every following step, and hence the subset containing both variables in the pair may appear in a late stage. In contrast, BOSS takes ordered predictors provided by FS, and re-orders them by performing BS upon their orthogonal basis, which gives a greater chance for both variables in the pair to appear early in the solution path of BOSS, and potentially results in a better predictive performance than FS. Furthermore, in this example, we note that BOSS provides a better solution path than BS until step $5$ (except the fourth step), and the two methods give similar performances in further steps.

\begin{figure}[ht!]
	\centering
	\includegraphics[width=\textwidth]{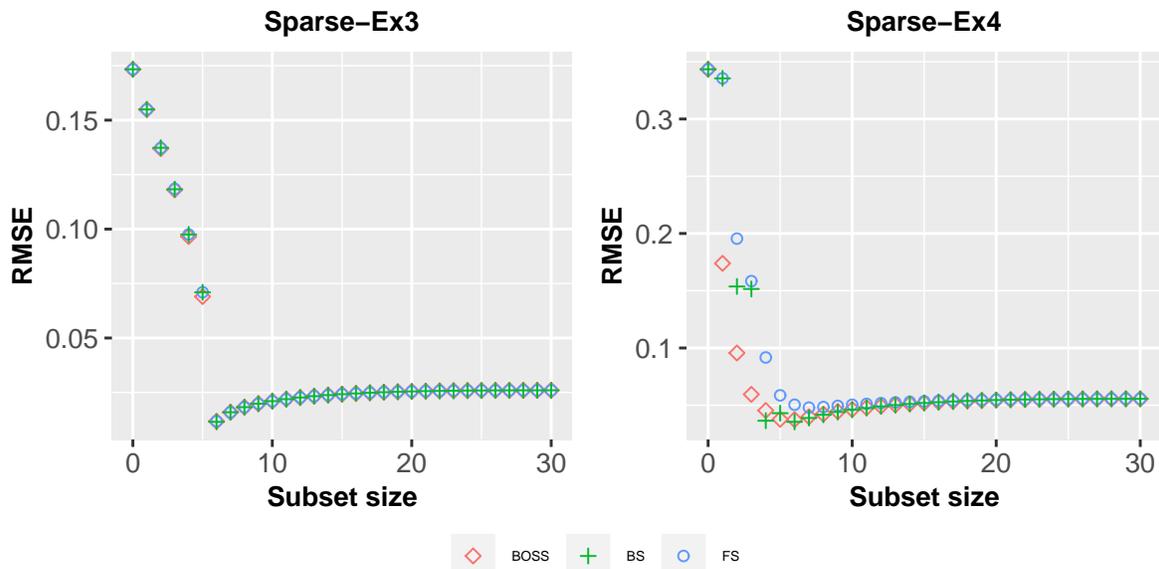}
	\caption{RMSE at each subset size, average over $1000$ replications. Note that for BOSS, the subset size $k_Q$ denotes the number of non-zero coefficients in $\tilde{\gamma}(k_Q)$. In both scenarios, we have $n=200$, $p=30$, $\rho=0.9$ and high SNR.}
	\label{fig:lossratio_fs_boss_k}
\end{figure}

\subsection{AICc-hdf as the selection rule for BOSS}

We apply AICc-hdf to choose the single optimal subset from the $K+1$ candidates. The implementation is discussed in step 6 of Algorithm \ref{alg:boss}. The hdf is calculated via Algorithm \ref{alg:hdf} based on the orthogonalized predictors $Q_{S_K}$. As to the estimation of $\mu$ and $\sigma$, if $n>p$, we use the OLS estimates based on the full model, i.e. $\hat{\mu}=Q_{S_{K}} z$, $\hat{\sigma}^2 = \lVert y-\hat{\mu} \rVert_2^2/(n-p)$. If $n \le p$, we use the estimates based on the lasso fit as discussed in \citet{reid2016study}, i.e. $\hat{\mu}=\hat{\mu}_l$ and $\hat{\sigma} = \sqrt{\lVert y-\hat{\mu}_l \rVert_2^2 / (n-\text{df}(\hat{\mu}_l)-1)}$, where $\hat{\mu}_l$ are the lasso fitted values based on 10-fold CV and df$(\hat{\mu}_l)$ is the corresponding number of non-zero coefficients in the lasso estimate. 

A numerical justification of using hdf is given in Figure \ref{fig:boss_aicc_hdf_kl}, where we compare averages of AICc-hdf and $\widehat{\text{Err}}_{\text{KL}}$ over $1000$ replications for BOSS under various true models. The sparse model (Sparse-Ex3) has $p_0=6$ predictors with non-zero coefficients, and all of the predictors in the dense model have non-zero coefficients ($p_0=p$). The correlation between predictors is $\rho=0.5$. We see that by using the sample average to represent the population mean, $E$(AICc-hdf) generally tracks the expected KL, $E(\text{Err}_{\text{KL}})$, reasonably well. Discrepancies can be observed at subset size $k<p_0$, where the set of true predictors is not entirely included in the model. The derivations of the classic AIC and AICc (both with ndf plugged in according to our notation) are based on an assumption that the true predictors are included in the model. In the situation where this assumption is violated, AICc will no longer be unbiased, and a similar conjecture can be made here for AICc in the context of BOSS. Last and most importantly, AICc-hdf yields the same average selected size as $\widehat{\text{Err}}_{\text{KL}}$ across all scenarios.

We find similar evidence in the Supplemental Material Figure \ref{fig:boss_cp_edf_hdf} that $E(\text{C}_p\text{-hdf})$ tracks the expected prediction error $E(\text{Err}_\text{SE})$ well in most cases, and they lead to the same average selected subset size; further supporting the use of hdf for BOSS. As discussed in Section \ref{sec:aicc_performance_bs}, using C$_p$ as the selection rule can perform poorly in practice because of the need to estimate $\sigma$. Evidence of a similar property when using C$_p$ as the selection rule for BOSS can be found in the Online Supplemental Material. For this reason, we prefer AICc in feasible versions of selection. 

\begin{figure}[ht!]
	\centering
	\includegraphics[width=0.9\textwidth]{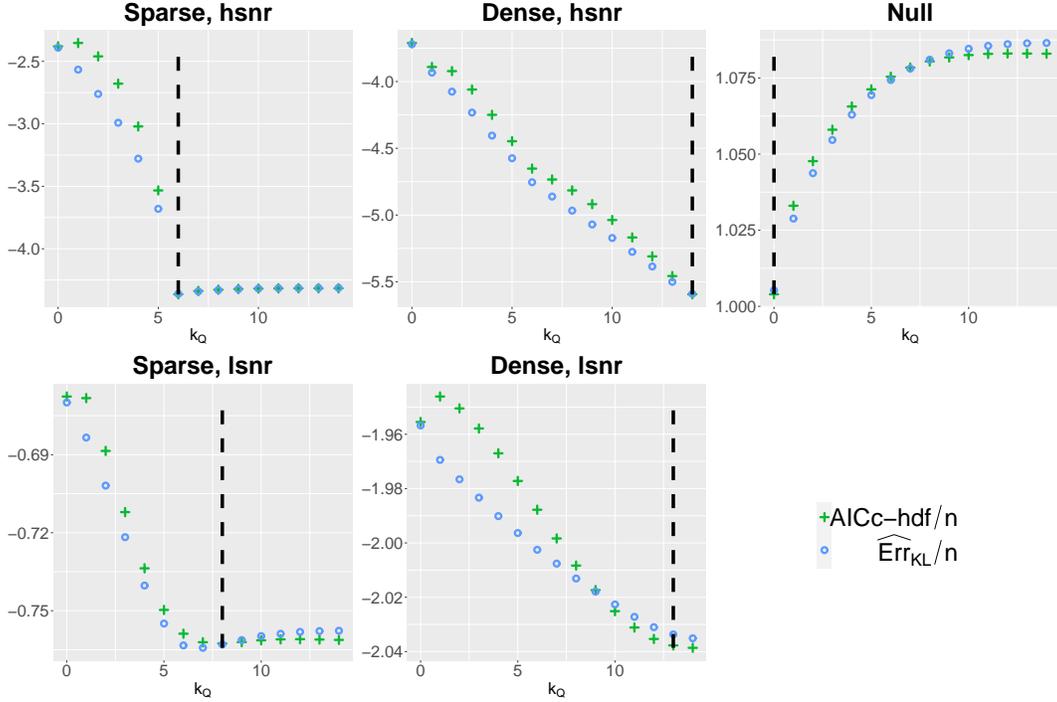}
	\caption{Averages of AICc-hdf and $\widehat{\text{Err}}_{\text{KL}}$ for BOSS over $1000$ replications. Here $X$ is general with $n=200$, $p=14$. Both criteria result in the same average of the selected subset size over the $1000$ replications (rounded to the nearest integer) as denoted by the dashed vertical lines. }
	\label{fig:boss_aicc_hdf_kl}
\end{figure}

\subsection{The performance of BOSS in simulations}
\label{sec:boss_performance}
We now study the performance of BOSS via simulations. Besides the above mentioned Sparse-Ex3, Sparse-Ex4 and Dense designs, we consider two additional sparse model designs that have different correlation structures between predictors. For all of the sparse examples, we take $p_0=6$. We further consider three levels of the correlations, $\rho \in [0,0.5,0.9]$, and twelve combinations of $(n,p)$, resulting in a total of $540$ configuration options. For each configuration, $1000$ replications are constructed and we present similar evaluation measures as introduced in Section \ref{sec:aicc_performance_bs}. One measure is the $\%$ worse than the best possible BOSS, where the best possible BOSS means that on a single fit, we choose the subset size $k_Q$ with the minimum RMSE among all $K+1$ candidates, as if an oracle tells us the best model. The details of the simulation setup are discussed in the Supplemental Material Section \ref{sec:simulation_setup_generalx}. The full set of results can be found in the Online Supplemental Material.

We looked at results using AICc-hdf, C$_p$-hdf and 10-fold CV for BOSS, and AICc-hdf performed the best (see Online Supplemental Material), so results for that method are presented here. For FS, we studied several information criteria that have been proposed in the literature, e.g. EBIC \citep{wang2009forward}, HDBIC and HDHQ \citep{ing2011stepwise}. None of these criteria provide strong enough penalties for larger subset sizes. For high dimensional problems ($n<p$), they tend to choose subsets with size close to $n$, which is far from the truth for sparse designs ($p_0=6$). To remedy this overfitting problem, \citet{ing2011stepwise} suggests using a stopping rule, and one example is to only consider subsets with size $k < 5\sqrt{n / \log(p)}$. We find that this ad-hoc stopping rule avoids the problem of not penalizing larger subsets sufficiently, and provides reasonably good performance for sparse designs. However, for dense designs, using the stopping rule gives subsets that substantially underfit. It is clear that using a fixed stopping rule is not appropriate for all problems. For this reason, we prefer 10-fold CV as the selection rule for FS, and we find it overall outperforms these suggested information criteria (see Online Supplemental Material). Finally, we fit BS via the ``leaps'' algorithm for the case where $p \le 30$, and use 10-fold CV as the selection rule. BOSS and FS are fitted using our {\tt{R}} package \pkg{BOSSReg}\footnote{\url{https://github.com/sentian/BOSSreg}. A stable version of the R package is available on \textit{CRAN}.}, and the BS is fitted using the {\tt{R}} package \pkg{leaps} \citep{FortrancodebyAlanMiller2020}.

We also consider some popular regularization methods, including lasso \citep{Tibshirani1996}, SparseNet \citep{Mazumder2011} and gamma lasso \citep{Taddy2017}. We use the {\tt{R}} packages \pkg{glmnet} \citep{Friedman2010}, \pkg{sparsenet} \citep{Mazumder2011}, and \pkg{gamlr} \citep{Taddy2017}, respectively, to fit them, which are all available on \textit{CRAN}. We also consider a simplified version of the relaxed lasso \citep{Meinshausen2007}, which was discussed in \citet{Hastie2017} and can be fitted using the {\tt{R}} package \pkg{bestsubset}\footnote{The package is available at https://github.com/ryantibs/best-subset. We appreciate Prof. Ryan Tibshirani for the suggestion of fitting the simplified relaxed lasso.}. As to the selection rule, we use AICc for lasso, and 10-fold CV for the rest. In addition to these selectors, we have also considered 10-fold CV for lasso. We find (in the Online Supplement) that 10-fold CV performs similarly to AICc for lasso. In fact, the use of AICc for lasso was explored in \citet{Flynn2013}, where the authors proved that AICc is asymptotically efficient while performing similarly to CV. We further notice (results given in the Online Supplement) that SparseNet generally performs better than the relaxed lasso and gamma lasso, and hence only the results for SparseNet will be presented here. 

Note that there is an extensive list of regression estimators existing in the literature and the list is still growing fast. For instance, recent studies by \citet{hazimeh2020fast} and \citet{bertsimas2020sparse} considered regularized best subset problems that combine a lasso or ridge penalty with the cardinality constraint (either in the Lagrangian or constrained form), and the authors provided fast solvers by using modern optimization tools. We expect the general conclusions below regarding the relative performance of BOSS and regularization methods hold for the new methods. Our simulation code \footnote{The code is available at \url{https://github.com/sentian/BOSSreg}.} is structured to be easily extendable, and we invite interested readers to perform further comparisons.

A selected set of simulation results is presented in Table \ref{tab:boss_regu} and \ref{tab:boss_regu_highdim}. Here is a brief summary of the results:

\begin{itemize}
	\item For BOSS, AICc-hdf has a significant advantage over CV in terms of predictive performance, except for low SNR and $n$ is small, in which case the selection rules are comparable. CV is also ten times heavier in terms of computation than AICc-hdf. These results are similar to the comparison of AICc-hdf and CV for BS with an orthogonal $X$ as discussed in Section \ref{sec:aicc_performance_bs}. Overall, the simulations indicate that AICc with hdf used in place of edf is a reasonable selection rule for an LS-based method that can be applied in practice without the requirement that the predictors are orthogonal. In the following discussions, when we refer to BOSS, we mean BOSS-AICc-hdf. 

	\item The performance of BOSS is comparable to the performance of BS when BS is feasible. With a small sample size $n=200$, BOSS performs either similarly to or better than BS for a high SNR, and it performs either similarly to or slightly worse than BS for a low SNR. With a large sample size $n=2000$, BOSS is generally better than BS. Furthermore, BOSS only requires fitting the procedure once while BS uses CV as the selection rule, and a single fit of BOSS only has computational cost $O(npK)$ so that BOSS is feasible for high dimensional problems.

	\item The performance of BOSS is generally better than the performance of FS. In the Dense model, and Sparse-Ex3 with $n=200$ and low SNR, we see that BOSS performs similarly to FS. In all other scenarios, the advantage of BOSS is obvious. For example, in Sparse-Ex4 with $n=200$, high SNR and $\rho=0.9$, FS is almost ten times worse than BOSS in terms of RMSE. Recall that Sparse-Ex4 is an example where FS has trouble stepping in all of the true predictors (with $\beta_j \ne 0$) in early steps. This is evidenced by the fact that FS chooses eight extra predictors on average in this situation, while BOSS only chooses approximately two extra predictors. Furthermore, FS based on CV is ten times computationally heavier than BOSS. 

	\item Compared to the regularization methods, with a small sample size $n=200$, BOSS is the best when SNR is high, lasso is the best when SNR is low and SparseNet is in between. The lasso has the property of ``over-select and shrink,'' in order to retain less bias on the large non-zero estimates. In a high SNR, this property can result in disastrous performance, especially when $p$ is large. For example, in Sparse-Ex3, high SNR, $\rho=0.5$ and $p=180$, the relative efficiency of lasso is only $0.43$ and it significantly overfits. However, this property can be beneficial when SNR is low, as a method like BS has higher chance to miss the true predictors (less sparsistency). With a large sample size $n=550$ and $n=2000$, BOSS is almost always the best even when SNR is low. 

	\item In terms of support recovery in the sparse true models, LS-based methods can recover the true predictors (those with $\beta_j \ne 0$) and rarely include any noise predictors (those with $\beta_j = 0$) when SNR is high or the sample size $n$ is large. However, SparseNet and lasso generally overfit, with the latter being worse in that regard. In the low SNR and small $n$ scenario, lasso and SparseNet have more opportunity to recover the true predictors, but it comes with a price of including more false positives. 

\end{itemize}

Given the spirit of the summary above, it is important to point out the relevant work of \citet{Hastie2017}, where the authors provide a comprehensive set of simulation comparisons on BS, lasso and relaxed lasso. The authors concluded that BS performs the best in high SNR, lasso is the best in low SNR while relaxed lasso is in between. Given the similarity we have noticed between BOSS and BS, it is not surprising that this coincides with our results for BOSS here when sample size is relatively small ($n=200$). However, we find BOSS to be the best for large sample size $n$ even when the SNR is low (note that \citet{Hastie2017} did not examine any sample sizes greater than $n=500$). Moreover, it should be noted that \citet{Hastie2017} focus on the best possible performance of each method by applying a separate validation set drawn from the true model, rather than on feasible selection, as is considered in this study.

\begin{table}[ht]
\centering
\caption{The performance of BOSS compared to other methods for $n>p$. Selection rules are 'AICc-hdf/CV' for BOSS, 
                AICc for lasso and CV for the remaining methods in the table, respectively.} 
\label{tab:boss_regu}
\scalebox{0.5}{
\begin{tabular}{|c|c|c|c|ccccc|ccccc|ccccc|}
  \toprule 
 \multicolumn{1}{|r}{} & \multicolumn{1}{r}{} & \multicolumn{1}{r}{} &       & \multicolumn{5}{c|}{Sprse-Ex3}        & \multicolumn{5}{c|}{Sparse-Ex4}       & \multicolumn{5}{c|}{Dense} \\
 \cmidrule{5-19}\multicolumn{1}{|r}{} & \multicolumn{1}{r}{} & \multicolumn{1}{r}{} &       & BOSS  & BS    & FS    & lasso & SparseNet & BOSS  & BS    & FS    & lasso & SparseNet & BOSS  & BS    & FS    & lasso & \multicolumn{1}{c|}{SparseNet}  \\
 \cmidrule{5-19}\multicolumn{1}{|c}{} & \multicolumn{1}{c}{} & \multicolumn{1}{c}{} &       & \multicolumn{15}{c|}{\% worse than the best possible BOSS}  \\
 \midrule 
 \multirow{8}[4]{*}{n=200} & \multirow{4}[2]{*}{hsnr} & \multirow{2}[1]{*}{$\rho=0.5$} & p=30 & 2/22 & 24 & 22 & 70 & 14 & 19/24 & 23 & 28 & 49 & 21 & 1/8 & 9 & 8 & 2 & 5 \\ 
   &  &  & p=180 & 1/18 & - & 21 & 135 & 17 & 4/15 & - & 16 & 82 & 13 & 14/13 & - & 16 & 47 & 8 \\ 
   &  & \multirow{2}[1]{*}{$\rho=0.9$} & p=30 & 5/41 & 17 & 41 & 66 & 12 & 21/33 & 21 & 56 & 73 & 28 & 2/9 & 10 & 8 & 2 & 8 \\ 
   &  &  & p=180 & 3/27 & - & 29 & 126 & 16 & 7/34 & - & 68 & 123 & -10 & 15/12 & - & 12 & 71 & 20 \\ 
  \cmidrule{2-19} & \multirow{4}[2]{*}{lsnr} & \multirow{2}[1]{*}{$\rho=0.5$} & p=30 & 30/25 & 25 & 25 & 0 & 7 & 35/32 & 23 & 34 & 35 & 23 & 18/17 & 16 & 18 & 10 & 11 \\ 
   &  &  & p=180 & 11/13 & - & 13 & -3 & 3 & 31/26 & - & 34 & 33 & 20 & 4/8 & - & 9 & 2 & 6 \\ 
   &  & \multirow{2}[1]{*}{$\rho=0.9$} & p=30 & 28/24 & 23 & 23 & -2 & 5 & 32/27 & 18 & 78 & 71 & 44 & 15/14 & 15 & 13 & 10 & 12 \\ 
   &  &  & p=180 & 16/16 & - & 15 & -5 & 1 & 17/18 & - & 36 & 33 & 35 & 12/10 & - & 10 & 15 & 9 \\ 
  \midrule \multirow{8}[4]{*}{n=2000} & \multirow{4}[2]{*}{hsnr} & \multirow{2}[1]{*}{$\rho=0.5$} & p=30 & 3/22 & 22 & 21 & 73 & 14 & 7/29 & 21 & 23 & 86 & 12 & 0/3 & 0 & 0 & 1 & 0 \\ 
   &  &  & p=180 & 1/22 & - & 22 & 130 & 14 & 6/28 & - & 21 & 174 & 12 & 8/9 & - & 12 & 37 & 8 \\ 
   &  & \multirow{2}[1]{*}{$\rho=0.9$} & p=30 & 2/21 & 21 & 22 & 74 & 13 & 32/33 & 16 & 33 & 108 & 12 & 0/3 & 1 & 1 & 2 & 1 \\ 
   &  &  & p=180 & 1/21 & - & 22 & 135 & 14 & 15/25 & - & 90 & 226 & 10 & 10/10 & - & 9 & 39 & 17 \\ 
  \cmidrule{2-19} & \multirow{4}[2]{*}{lsnr} & \multirow{2}[1]{*}{$\rho=0.5$} & p=30 & 5/22 & 21 & 21 & 73 & 14 & 13/30 & 22 & 23 & 61 & 15 & 2/9 & 10 & 10 & 2 & 7 \\ 
   &  &  & p=180 & 5/22 & - & 22 & 129 & 13 & 8/27 & - & 20 & 125 & 10 & 11/13 & - & 16 & 32 & 10 \\ 
   &  & \multirow{2}[1]{*}{$\rho=0.9$} & p=30 & 5/21 & 20 & 21 & 53 & 3 & 27/34 & 16 & 40 & 85 & 12 & 3/11 & 11 & 9 & 3 & 8 \\ 
   &  &  & p=180 & 4/17 & - & 17 & 92 & -5 & 14/27 & - & 104 & 179 & 20 & 12/13 & - & 11 & 40 & 18 \\ 
   \midrule 
 \multicolumn{1}{|c}{} & \multicolumn{1}{c}{} & \multicolumn{1}{c}{} &       & \multicolumn{15}{c|}{Relative efficiency} \\
 \midrule 
\multirow{8}[4]{*}{n=200} & \multirow{4}[2]{*}{hsnr} & \multirow{2}[1]{*}{$\rho=0.5$} & p=30 & 1/0.84 & 0.82 & 0.84 & 0.6 & 0.9 & 1/0.96 & 0.97 & 0.93 & 0.8 & 0.99 & 0.98/0.93 & 0.91 & 0.93 & 0.98 & 0.94 \\ 
   &  &  & p=180 & 1/0.85 & - & 0.84 & 0.43 & 0.86 & 1/0.91 & - & 0.9 & 0.57 & 0.92 & 0.95/0.96 & - & 0.93 & 0.73 & 1 \\ 
   &  & \multirow{2}[1]{*}{$\rho=0.9$} & p=30 & 1/0.74 & 0.9 & 0.74 & 0.63 & 0.93 & 1/0.91 & 1 & 0.78 & 0.7 & 0.95 & 0.98/0.91 & 0.9 & 0.92 & 0.98 & 0.93 \\ 
   &  &  & p=180 & 1/0.8 & - & 0.79 & 0.45 & 0.88 & 0.84/0.68 & - & 0.54 & 0.41 & 1 & 0.97/1 & - & 1 & 0.65 & 0.93 \\ 
  \cmidrule{2-19} & \multirow{4}[2]{*}{lsnr} & \multirow{2}[1]{*}{$\rho=0.5$} & p=30 & 0.77/0.8 & 0.8 & 0.8 & 1 & 0.93 & 0.91/0.93 & 1 & 0.92 & 0.91 & 1 & 0.93/0.94 & 0.95 & 0.93 & 1 & 0.99 \\ 
   &  &  & p=180 & 0.87/0.86 & - & 0.86 & 1 & 0.94 & 0.92/0.96 & - & 0.9 & 0.91 & 1 & 0.97/0.94 & - & 0.94 & 1 & 0.96 \\ 
   &  & \multirow{2}[1]{*}{$\rho=0.9$} & p=30 & 0.76/0.79 & 0.8 & 0.8 & 1 & 0.93 & 0.89/0.93 & 1 & 0.66 & 0.69 & 0.82 & 0.96/0.97 & 0.96 & 0.97 & 1 & 0.98 \\ 
   &  &  & p=180 & 0.82/0.82 & - & 0.83 & 1 & 0.95 & 1/0.99 & - & 0.86 & 0.88 & 0.87 & 0.97/0.99 & - & 1 & 0.95 & 1 \\ 
  \midrule \multirow{8}[4]{*}{n=2000} & \multirow{4}[2]{*}{hsnr} & \multirow{2}[1]{*}{$\rho=0.5$} & p=30 & 1/0.85 & 0.85 & 0.85 & 0.59 & 0.91 & 1/0.83 & 0.89 & 0.87 & 0.58 & 0.96 & 0.98/0.96 & 0.98 & 0.98 & 0.97 & 0.98 \\ 
   &  &  & p=180 & 1/0.83 & - & 0.83 & 0.44 & 0.89 & 1/0.83 & - & 0.88 & 0.39 & 0.95 & 1/0.99 & - & 0.96 & 0.79 & 1 \\ 
   &  & \multirow{2}[1]{*}{$\rho=0.9$} & p=30 & 1/0.84 & 0.85 & 0.84 & 0.59 & 0.9 & 0.84/0.84 & 0.96 & 0.84 & 0.54 & 1 & 1/0.97 & 0.99 & 0.99 & 0.98 & 0.99 \\ 
   &  &  & p=180 & 1/0.84 & - & 0.83 & 0.43 & 0.89 & 0.96/0.88 & - & 0.58 & 0.34 & 1 & 0.99/0.99 & - & 1 & 0.78 & 0.93 \\ 
  \cmidrule{2-19} & \multirow{4}[2]{*}{lsnr} & \multirow{2}[1]{*}{$\rho=0.5$} & p=30 & 1/0.86 & 0.86 & 0.86 & 0.61 & 0.92 & 1/0.87 & 0.93 & 0.92 & 0.7 & 0.99 & 0.98/0.91 & 0.9 & 0.91 & 0.98 & 0.93 \\ 
   &  &  & p=180 & 1/0.86 & - & 0.86 & 0.46 & 0.93 & 1/0.85 & - & 0.9 & 0.48 & 0.98 & 1/0.97 & - & 0.95 & 0.83 & 1 \\ 
   &  & \multirow{2}[1]{*}{$\rho=0.9$} & p=30 & 0.98/0.85 & 0.86 & 0.85 & 0.67 & 1 & 0.88/0.83 & 0.97 & 0.8 & 0.61 & 1 & 1/0.92 & 0.92 & 0.94 & 1 & 0.95 \\ 
   &  &  & p=180 & 0.91/0.81 & - & 0.81 & 0.49 & 1 & 1/0.9 & - & 0.56 & 0.41 & 0.95 & 0.99/0.99 & - & 1 & 0.8 & 0.94 \\ 
   \midrule 
 \multicolumn{1}{|c}{} & \multicolumn{1}{c}{} & \multicolumn{1}{c}{} &       & \multicolumn{15}{c|}{Sparsistency (number of extra variables)} \\
 \midrule 
\multirow{8}[4]{*}{n=200} & \multirow{4}[2]{*}{hsnr} & \multirow{2}[1]{*}{$\rho=0.5$} & p=30 & 6(0)/6(0.6) & 6(0.7) & 6(0.6) & 6(7.9) & 6(1.1) & 4.4(0.2)/5(1) & 5(1) & 4.8(1.1) & 5.7(10.4) & 4.8(2.1) & 29.6/26.1 & 25.1 & 26 & 29.1 & 27 \\ 
   &  &  & p=180 & 6(0)/6(0.3) & - & 6(0.4) & 6(16.6) & 6(2.4) & 4(0)/4.2(0.5) & - & 4.1(0.5) & 5.1(20.2) & 4.2(3.5) & 17/20.2 & - & 19.6 & 52.2 & 32.4 \\ 
   &  & \multirow{2}[1]{*}{$\rho=0.9$} & p=30 & 6(0.6)/6(2.1) & 6(0.8) & 6(2.1) & 6(9.2) & 6(1.6) & 5.1(2.8)/5.3(3.8) & 5(1) & 4.8(4.1) & 5.8(17.8) & 4.4(2.7) & 29.3/25.2 & 23 & 24.6 & 28.9 & 26.2 \\ 
   &  &  & p=180 & 6(0.1)/6(0.6) & - & 6(0.6) & 6(16.2) & 6(2.4) & 4.2(2.4)/4.3(4.3) & - & 4.3(8) & 4.6(44.2) & 4.1(3.1) & 15.6/21.3 & - & 17.2 & 54.4 & 37.7 \\ 
  \cmidrule{2-19} & \multirow{4}[2]{*}{lsnr} & \multirow{2}[1]{*}{$\rho=0.5$} & p=30 & 2.9(2)/3.6(2.4) & 3.4(2.1) & 3.5(2.3) & 5.1(6.9) & 4.7(5.3) & 2.3(1)/2.7(1.3) & 2.6(1) & 2.6(1.5) & 3.6(6.9) & 2.8(2.9) & 5.7/7.5 & 6.6 & 7.1 & 5.3 & 10.3 \\ 
   &  &  & p=180 & 0.3(0.1)/1(0.7) & - & 1(0.7) & 3(9.7) & 2.6(9.3) & 1(0.2)/1.6(0.9) & - & 1.3(0.8) & 2.2(10.7) & 2.1(6.5) & 0.2/1.1 & - & 0.9 & 2.7 & 6.9 \\ 
   &  & \multirow{2}[1]{*}{$\rho=0.9$} & p=30 & 1.9(2.3)/2.4(3) & 2.5(2.8) & 2.4(2.9) & 3.9(7.5) & 3.7(6.1) & 2.7(3.9)/3.2(5) & 2.7(0.9) & 2.2(4.4) & 3(10.1) & 2.9(8) & 4.1/5.2 & 4.3 & 4.5 & 8.4 & 5.9 \\ 
   &  &  & p=180 & 0.5(0.2)/1.1(1.1) & - & 1.1(1.1) & 3.2(11.1) & 3(10.8) & 0.7(1.7)/1.1(5.5) & - & 0.2(0.6) & 0.3(4.8) & 0.6(9) & 1/2.1 & - & 1.6 & 3.9 & 3.3 \\ 
  \midrule \multirow{8}[4]{*}{n=2000} & \multirow{4}[2]{*}{hsnr} & \multirow{2}[1]{*}{$\rho=0.5$} & p=30 & 6(0.1)/6(0.6) & 6(0.6) & 6(0.6) & 6(8.4) & 6(1) & 6(0.2)/6(0.6) & 6(0.6) & 6(0.6) & 6(10.9) & 6(0.8) & 30/30 & 30 & 30 & 30 & 30 \\ 
   &  &  & p=180 & 6(0)/6(0.4) & - & 6(0.4) & 6(21.5) & 6(2.3) & 6(0.1)/6(0.3) & - & 6(0.3) & 6(32.2) & 6(1.3) & 34.5/35.1 & - & 32.6 & 106.5 & 43 \\ 
   &  & \multirow{2}[1]{*}{$\rho=0.9$} & p=30 & 6(0)/6(0.6) & 6(0.6) & 6(0.6) & 6(9.2) & 6(1) & 6(0.4)/6(0.7) & 6(0.6) & 6(1.3) & 6(17.8) & 6(1.6) & 30/29.9 & 29.9 & 29.9 & 30 & 30 \\ 
   &  &  & p=180 & 6(0)/6(0.4) & - & 6(0.4) & 6(23.2) & 6(2.2) & 6(1.4)/6(1.7) & - & 5.9(3.8) & 6(72.7) & 6(8.7) & 35/38.6 & - & 30.2 & 109.6 & 52.4 \\ 
  \cmidrule{2-19} & \multirow{4}[2]{*}{lsnr} & \multirow{2}[1]{*}{$\rho=0.5$} & p=30 & 6(0.1)/6(0.6) & 6(0.6) & 6(0.6) & 6(8.3) & 6(0.7) & 4.2(0.4)/4.3(0.7) & 4.3(0.6) & 4.2(0.7) & 5.2(9.6) & 4.3(1) & 29/22.7 & 21.2 & 22.1 & 28 & 24.1 \\ 
   &  &  & p=180 & 6(0.1)/6(0.4) & - & 6(0.4) & 6(21.2) & 6(0.9) & 4(0.1)/4(0.4) & - & 4(0.4) & 4.6(26.1) & 4.1(1.3) & 16/17 & - & 14.3 & 61.8 & 25.3 \\ 
   &  & \multirow{2}[1]{*}{$\rho=0.9$} & p=30 & 5.8(0.3)/5.8(1.1) & 5.8(1.1) & 5.8(1.1) & 6(9.2) & 6(0.8) & 4.4(1.9)/4.4(1.7) & 4.3(0.6) & 4.3(2.2) & 5.4(16.6) & 4.2(2.5) & 28.8/21.2 & 18.5 & 20.2 & 27.6 & 23.5 \\ 
   &  &  & p=180 & 5.7(0.3)/5.7(0.7) & - & 5.7(0.7) & 6(23) & 6(1) & 4.1(3.6)/4.1(4.4) & - & 3.7(3.7) & 4.6(60.3) & 4.2(14.2) & 16.6/21.4 & - & 11.8 & 65.3 & 32.3 \\ 
   \bottomrule 
\end{tabular}
}
\end{table}

\begin{table}[ht]
\centering
\caption{The performance of BOSS compared to other methods for $n \le p$. Selection rules are for 'AICc-hdf/CV' for BOSS, 
                AICc for lasso and CV for the remaining methods in the table, respectively.} 
\label{tab:boss_regu_highdim}
\scalebox{0.52}{
\begin{tabular}{|c|c|c|c|ccccc|ccccc|ccccc|}
  \toprule 
 \multicolumn{1}{|r}{} & \multicolumn{1}{r}{} & \multicolumn{1}{r}{} &       & \multicolumn{5}{c|}{Sprse-Ex3}        & \multicolumn{5}{c|}{Sparse-Ex4}       & \multicolumn{5}{c|}{Dense} \\
 \cmidrule{5-19}\multicolumn{1}{|r}{} & \multicolumn{1}{r}{} & \multicolumn{1}{r}{} &       & BOSS  & BS & FS    & lasso    & SparseNet & BOSS  & BS & FS    & lasso    & SparseNet & BOSS  & BS & FS    & lasso  & \multicolumn{1}{c|}{SparseNet}  \\
 \cmidrule{5-19}\multicolumn{1}{|c}{} & \multicolumn{1}{c}{} & \multicolumn{1}{c}{} &       & \multicolumn{15}{c|}{\% worse than the best possible BOSS}  \\
 \midrule 
 \multirow{8}[4]{*}{n=200} & \multirow{4}[2]{*}{hsnr} & \multirow{2}[1]{*}{$\rho=0.5$} & p=550 & 2/19 & - & 20 & 168 & 25 & 4/10 & - & 11 & 94 & 12 & 7/13 & - & 13 & 70 & 3 \\ 
   &  &  & p=1000 & 9/18 & - & 18 & 176 & 25 & 8/10 & - & 11 & 103 & 13 & 7/15 & - & 17 & 113 & 3 \\ 
   &  & \multirow{2}[1]{*}{$\rho=0.9$} & p=550 & 6/48 & - & 48 & 150 & 38 & 9/33 & - & 67 & 186 & -9 & 14/14 & - & 17 & 131 & 17 \\ 
   &  &  & p=1000 & 14/37 & - & 38 & 167 & 22 & 16/99 & - & 122 & 198 & -21 & 11/14 & - & 18 & 158 & 125 \\ 
  \cmidrule{2-19} & \multirow{4}[2]{*}{lsnr} & \multirow{2}[1]{*}{$\rho=0.5$} & p=550 & 10/9 & - & 9 & -4 & 3 & 23/21 & - & 28 & 27 & 20 & 5/6 & - & 6 & 1 & 4 \\ 
   &  &  & p=1000 & 15/10 & - & 10 & -2 & 4 & 19/15 & - & 19 & 15 & 15 & 6/6 & - & 6 & 0 & 4 \\ 
   &  & \multirow{2}[1]{*}{$\rho=0.9$} & p=550 & 10/10 & - & 10 & -4 & 2 & 6/8 & - & 11 & 9 & 11 & 18/13 & - & 12 & 13 & 11 \\ 
   &  &  & p=1000 & 15/8 & - & 8 & -4 & 2 & 5/5 & - & 7 & 3 & 6 & 15/12 & - & 12 & 12 & 9 \\ 
  \midrule \multirow{8}[4]{*}{n=500} & \multirow{4}[2]{*}{hsnr} & \multirow{2}[1]{*}{$\rho=0.5$} & p=550 & 0/20 & - & 22 & 157 & 19 & 31/26 & - & 37 & 97 & 30 & 9/10 & - & 13 & 57 & 8 \\ 
   &  &  & p=1000 & 1/21 & - & 22 & 172 & 21 & 29/25 & - & 36 & 114 & 29 & 8/11 & - & 14 & 85 & 9 \\ 
   &  & \multirow{2}[1]{*}{$\rho=0.9$} & p=550 & 1/22 & - & 23 & 158 & 19 & 12/18 & - & 30 & 162 & 29 & 13/10 & - & 14 & 79 & 23 \\ 
   &  &  & p=1000 & 2/22 & - & 23 & 175 & 21 & 8/13 & - & 21 & 172 & 19 & 12/11 & - & 14 & 107 & 33 \\ 
  \cmidrule{2-19} & \multirow{4}[2]{*}{lsnr} & \multirow{2}[1]{*}{$\rho=0.5$} & p=550 & 34/28 & - & 28 & 23 & 21 & 24/31 & - & 33 & 71 & 28 & 16/12 & - & 14 & 17 & 11 \\ 
   &  &  & p=1000 & 25/26 & - & 26 & 21 & 18 & 17/25 & - & 26 & 67 & 20 & 14/12 & - & 15 & 17 & 11 \\ 
   &  & \multirow{2}[1]{*}{$\rho=0.9$} & p=550 & 30/25 & - & 26 & 18 & 17 & 12/14 & - & 73 & 82 & 34 & 6/7 & - & 8 & 14 & 10 \\ 
   &  &  & p=1000 & 25/26 & - & 25 & 17 & 15 & 6/10 & - & 41 & 44 & 28 & 7/7 & - & 8 & 17 & 12 \\ 
   \midrule 
 \multicolumn{1}{|c}{} & \multicolumn{1}{c}{} & \multicolumn{1}{c}{} &       & \multicolumn{15}{c|}{Relative efficiency} \\
 \midrule 
\multirow{8}[4]{*}{n=200} & \multirow{4}[2]{*}{hsnr} & \multirow{2}[1]{*}{$\rho=0.5$} & p=550 & 1/0.85 & - & 0.84 & 0.38 & 0.81 & 1/0.95 & - & 0.93 & 0.54 & 0.92 & 0.96/0.91 & - & 0.91 & 0.6 & 1 \\ 
   &  &  & p=1000 & 1/0.92 & - & 0.92 & 0.39 & 0.87 & 1/0.99 & - & 0.98 & 0.53 & 0.96 & 0.96/0.9 & - & 0.88 & 0.48 & 1 \\ 
   &  & \multirow{2}[1]{*}{$\rho=0.9$} & p=550 & 1/0.72 & - & 0.72 & 0.43 & 0.77 & 0.84/0.68 & - & 0.55 & 0.32 & 1 & 1/1 & - & 0.98 & 0.49 & 0.97 \\ 
   &  &  & p=1000 & 1/0.83 & - & 0.82 & 0.43 & 0.93 & 0.68/0.4 & - & 0.36 & 0.27 & 1 & 1/0.97 & - & 0.94 & 0.43 & 0.49 \\ 
  \cmidrule{2-19} & \multirow{4}[2]{*}{lsnr} & \multirow{2}[1]{*}{$\rho=0.5$} & p=550 & 0.88/0.88 & - & 0.88 & 1 & 0.94 & 0.97/0.99 & - & 0.94 & 0.94 & 1 & 0.96/0.95 & - & 0.95 & 1 & 0.97 \\ 
   &  &  & p=1000 & 0.85/0.89 & - & 0.89 & 1 & 0.94 & 0.96/1 & - & 0.96 & 1 & 1 & 0.95/0.95 & - & 0.95 & 1 & 0.97 \\ 
   &  & \multirow{2}[1]{*}{$\rho=0.9$} & p=550 & 0.87/0.87 & - & 0.87 & 1 & 0.94 & 1/0.98 & - & 0.95 & 0.97 & 0.95 & 0.95/0.99 & - & 0.99 & 0.98 & 1 \\ 
   &  &  & p=1000 & 0.84/0.89 & - & 0.89 & 1 & 0.94 & 0.98/0.98 & - & 0.97 & 1 & 0.98 & 0.94/0.98 & - & 0.98 & 0.98 & 1 \\ 
  \midrule \multirow{8}[4]{*}{n=500} & \multirow{4}[2]{*}{hsnr} & \multirow{2}[1]{*}{$\rho=0.5$} & p=550 & 1/0.84 & - & 0.83 & 0.39 & 0.85 & 0.96/1 & - & 0.92 & 0.64 & 0.97 & 0.99/0.98 & - & 0.95 & 0.69 & 1 \\ 
   &  &  & p=1000 & 1/0.84 & - & 0.83 & 0.37 & 0.84 & 0.97/1 & - & 0.91 & 0.58 & 0.97 & 1/0.97 & - & 0.94 & 0.58 & 0.99 \\ 
   &  & \multirow{2}[1]{*}{$\rho=0.9$} & p=550 & 1/0.82 & - & 0.82 & 0.39 & 0.84 & 1/0.95 & - & 0.87 & 0.43 & 0.87 & 0.98/1 & - & 0.97 & 0.62 & 0.89 \\ 
   &  &  & p=1000 & 1/0.83 & - & 0.82 & 0.37 & 0.84 & 1/0.95 & - & 0.89 & 0.4 & 0.91 & 0.99/1 & - & 0.97 & 0.54 & 0.84 \\ 
  \cmidrule{2-19} & \multirow{4}[2]{*}{lsnr} & \multirow{2}[1]{*}{$\rho=0.5$} & p=550 & 0.9/0.94 & - & 0.95 & 0.98 & 1 & 1/0.95 & - & 0.93 & 0.73 & 0.97 & 0.96/0.99 & - & 0.97 & 0.95 & 1 \\ 
   &  &  & p=1000 & 0.94/0.93 & - & 0.93 & 0.97 & 1 & 1/0.94 & - & 0.93 & 0.7 & 0.97 & 0.98/0.99 & - & 0.97 & 0.95 & 1 \\ 
   &  & \multirow{2}[1]{*}{$\rho=0.9$} & p=550 & 0.9/0.93 & - & 0.93 & 0.99 & 1 & 1/0.98 & - & 0.65 & 0.61 & 0.84 & 1/0.99 & - & 0.98 & 0.93 & 0.97 \\ 
   &  &  & p=1000 & 0.92/0.92 & - & 0.92 & 0.98 & 1 & 1/0.97 & - & 0.76 & 0.74 & 0.83 & 1/1 & - & 1 & 0.91 & 0.96 \\ 
   \midrule 
 \multicolumn{1}{|c}{} & \multicolumn{1}{c}{} & \multicolumn{1}{c}{} &       & \multicolumn{15}{c|}{Sparsistency (number of extra variables)} \\
 \midrule 
\multirow{8}[4]{*}{n=200} & \multirow{4}[2]{*}{hsnr} & \multirow{2}[1]{*}{$\rho=0.5$} & p=550 & 6(0)/6(0.3) & - & 6(0.3) & 6(19.7) & 6(4.1) & 4(0.1)/4.1(0.3) & - & 4.1(0.3) & 4.7(23.2) & 4.3(5.3) & 16.5/15.9 & - & 15.7 & 52.8 & 29.4 \\ 
   &  &  & p=1000 & 6(0.3)/6(0.2) & - & 6(0.2) & 6(20.4) & 6(4.6) & 4.1(0.5)/4.1(0.3) & - & 4.1(0.3) & 4.4(24.3) & 4.2(5.4) & 17.9/15 & - & 14.6 & 45.9 & 34.5 \\ 
   &  & \multirow{2}[1]{*}{$\rho=0.9$} & p=550 & 6(0.6)/6(1.4) & - & 6(1.3) & 6(20.5) & 6(9.3) & 4(2)/4(2.9) & - & 4(7.1) & 4.1(61.2) & 4.1(4.7) & 13.3/15.6 & - & 13.1 & 32.8 & 45.1 \\ 
   &  &  & p=1000 & 6(0.6)/6(0.8) & - & 6(0.8) & 6(20.8) & 6(5.3) & 3.8(6.3)/3.8(5.7) & - & 3.5(10.5) & 4(67) & 4(6.6) & 14.6/17.2 & - & 16.2 & 7.6 & 22.3 \\ 
  \cmidrule{2-19} & \multirow{4}[2]{*}{lsnr} & \multirow{2}[1]{*}{$\rho=0.5$} & p=550 & 0.4(0.3)/0.7(0.4) & - & 0.7(0.4) & 2.6(11.6) & 2.4(13.8) & 0.9(0.4)/1.1(0.7) & - & 0.8(0.6) & 1.5(11.6) & 1.6(10.1) & 0.2/0.6 & - & 0.6 & 8.3 & 8.3 \\ 
   &  &  & p=1000 & 0.6(1.1)/0.5(0.4) & - & 0.6(0.4) & 2.2(10.9) & 2.1(14.2) & 0.8(1.2)/0.6(0.5) & - & 0.4(0.4) & 1(9.3) & 1.2(10.7) & 0.5/0.5 & - & 0.5 & 8 & 8.4 \\ 
   &  & \multirow{2}[1]{*}{$\rho=0.9$} & p=550 & 0.4(0.5)/0.5(0.7) & - & 0.5(0.7) & 2.1(12.3) & 2.1(14.5) & 0.3(0.4)/0.3(0.6) & - & 0.1(0.3) & 0.2(6.7) & 0.2(7.7) & 0.8/1.5 & - & 1.4 & 7.8 & 4.2 \\ 
   &  &  & p=1000 & 0.5(1.5)/0.4(0.4) & - & 0.4(0.4) & 2(11.3) & 1.9(15.3) & 0.1(0.5)/0.1(0.4) & - & 0(0.3) & 0.1(6.3) & 0.1(6.6) & 0.8/1.2 & - & 1.2 & 9 & 5.1 \\ 
  \midrule \multirow{8}[4]{*}{n=500} & \multirow{4}[2]{*}{hsnr} & \multirow{2}[1]{*}{$\rho=0.5$} & p=550 & 6(0)/6(0.3) & - & 6(0.3) & 6(24.1) & 6(3.4) & 4.6(0.1)/5.1(0.8) & - & 4.8(0.6) & 5.7(31.8) & 5.1(6.5) & 21.8/23.1 & - & 21.3 & 109.5 & 34.3 \\ 
   &  &  & p=1000 & 6(0)/6(0.3) & - & 6(0.3) & 6(26.6) & 6(4.1) & 4.7(0.2)/5.1(0.8) & - & 4.7(0.5) & 5.4(37.2) & 5(7.1) & 22.2/21.8 & - & 20.4 & 99.8 & 35.7 \\ 
   &  & \multirow{2}[1]{*}{$\rho=0.9$} & p=550 & 6(0)/6(0.3) & - & 6(0.3) & 6(24.5) & 6(3.4) & 4.4(1.6)/4.5(4) & - & 4.1(0.4) & 4.7(66.3) & 4.1(4.2) & 20.8/27 & - & 19.8 & 131.9 & 55.4 \\ 
   &  &  & p=1000 & 6(0)/6(0.3) & - & 6(0.3) & 6(27.6) & 6(4) & 4.3(0.9)/4.3(1.9) & - & 4(0.4) & 4.4(81.5) & 4.1(4.7) & 19.6/22.3 & - & 17.7 & 147.5 & 61.2 \\ 
  \cmidrule{2-19} & \multirow{4}[2]{*}{lsnr} & \multirow{2}[1]{*}{$\rho=0.5$} & p=550 & 3.5(0.3)/4.1(1) & - & 4.1(0.9) & 5.8(23.4) & 5.2(12.1) & 2.7(0.3)/2.9(0.6) & - & 2.7(0.6) & 3.7(24.2) & 2.9(4.6) & 2.1/4.2 & - & 2.9 & 22.4 & 13 \\ 
   &  &  & p=1000 & 3.6(0.5)/3.8(0.9) & - & 3.8(0.9) & 5.6(25.1) & 4.9(13.6) & 2.5(0.4)/2.5(0.6) & - & 2.3(0.5) & 3.1(26.3) & 2.6(4.8) & 2.2/3.7 & - & 2.5 & 19.2 & 13.2 \\ 
   &  & \multirow{2}[1]{*}{$\rho=0.9$} & p=550 & 3(0.9)/3.5(1.7) & - & 3.5(1.7) & 5.3(23.7) & 4.6(13.5) & 2.3(10.3)/2.5(13.9) & - & 0.8(1.2) & 1.3(21.2) & 2.2(34.3) & 1.5/3.9 & - & 2 & 10.4 & 4 \\ 
   &  &  & p=1000 & 3(1.1)/3.1(1.4) & - & 3.1(1.4) & 5.1(25.9) & 4.4(15) & 1.4(8)/1.5(10.4) & - & 0.3(0.7) & 0.5(12.5) & 1.3(37.3) & 1.6/2.9 & - & 2 & 10.9 & 3.3 \\ 
   \bottomrule 
\end{tabular}
}
\end{table}

\subsection{The performance of BOSS in real data analysis}
\label{sec:real_data}

We implement BOSS on five real datasets. We consider four datasets from the StatLib library\footnote{http://lib.stat.cmu.edu/datasets/}, which is maintained at Carnegie Mellon University. The ``Housing'' data are often used in comparisons of different regression methods. The aim is to predict the housing values in the suburbs of Boston based on $13$ predictors, including crime rate, property tax rate, pupil-teacher ratio, etc. The ``Hitters'' data contain the 1987 annual salary for MLB players. For each player, it records $19$ different performance metrics happening in 1986, such as number of times at bat, number of hits, etc., and the task is to predict the salary based on these statistics. The ``Auto'' data are driven by prediction of the miles per gallon of vehicles based on features like the horsepower, weight, etc. The ``College'' data contain various statistics for a large number of US colleges from the 1995 issue of ``US News and World Report'', and we use these statistics to predict the number of applications received. We also consider a dataset from the Machine Learning Repository\footnote{https://archive.ics.uci.edu/ml} that is maintained by UC Irvine. The ``ForestFire'' data are provided by \citet{cortez2007data} and the aim is to use meteorological and other data to predict the burned area of forest fires that happened in the northeast region of Portugal. The authors considered several machine learning algorithms, e.g. support vector regression, and concluded that the best prediction in terms of RMSE is the naive mean vector.


In real data analysis, one almost always would consider an intercept term. The way that BOSS handles the intercept term is described in steps 5-6 of Algorithm \ref{alg:boss}. Specifically, we first center both $X$ and $y$, and fit BOSS using AICc-hdf without an intercept to get $\hat{\beta}$. Then we calculate the intercept by $
\hat{\beta}_0=\bar{y} - \bar{X}^T \hat{\beta}$, which can be easily shown to be equivalent to fitting an intercept in every subset considered by BOSS. 

We compare the performance of BOSS with LS-based methods BS and FS, and with regularization methods lasso and SparseNet. All of the methods are fitted with an intercept term. Note that for the Forest Fires dataset, we fit BS via MIO \citep{Bertsimas2016} using the {\tt{R}} package \pkg{bestsubset} \citep{Hastie2017}, where we restrict subset size $k=0,\dots,10$, with $3$ minutes as the time budget to find an optimal solution for each $k$, as suggested by the authors. For all of the other datasets, BS is fitted using the \pkg{leaps} package. To measure the performance of each method, we apply the leave-one-out (LOO) testing procedure, in which we fit the method on all observations except one, test the performance on that particular observation, and repeat the procedure for all $n$ observations. 

Table \ref{tab:realdata} presents the average RMSE, the average number of predictors and average running time for various methods given by LOO testing. We see that BOSS provides the best predictive performance in all datasets except the ``Housing'' and ``College'' data where lasso is the best for those datasets and its RMSE is $0.3\%$ and $0.04\%$ lower than those of BOSS, respectively. Due to a highly optimized implementation of the cyclical coordinate descent, the ``glmnet'' algorithm is extremely fast in providing the lasso solution. BS is still not scalable to large dimensions, even by using the modern tools. With the dimension $p=55$, it takes around $350$ seconds to perform 10-fold CV for subset sizes restricted to be no greater than $10$. However, we observe that BOSS is reasonably computationally efficient and much faster than BS, FS and SparseNet. 

\begin{table}[ht]
\centering
\caption{Performance of subset selection methods on real datasets. 
               The results are averages of leave-one-out (LOO) testing. 
               The selection rules are AICc-hdf for BOSS, AICc for lasso and 10-fold CV for the rest, respectively. 
               The intercept term is always fitted and is not counted in the number of predictors. 
               Minimal values for the metrics for each dataset are given in bold face.} 
\label{tab:realdata}
\scalebox{0.9}{
\begin{tabular}{|c|c|ccc|c|c|}
  \toprule 
 Dataset (n, p) & Metrics & BOSS  & BS    & FS    & lasso & SparseNet  \\
 \midrule\multirow{3}[2]{*}{Housing (506, 13)} & RMSE & 3.372 & 3.37 & 3.383 & \textbf{3.363} & 3.369 \\ 
   & \# predictors & 12.004 & \textbf{12.002} & 12.026 & 12.012 & \textbf{12.002} \\ 
   & running time (s) & 0.021 & 0.066 & 0.156 & 0.007 & 0.39 \\ 
  \midrule\multirow{3}[2]{*}{Hitters (263, 19)} & RMSE & \textbf{233.853} & 236.989 & 238.222 & 234.064 & 238.375 \\ 
   & \# predictors & 11.152 & 10.852 & \textbf{10.662} & 14.205 & 12.51 \\ 
   & running time (s) & 0.014 & 0.095 & 0.104 & 0.008 & 0.493 \\ 
  \midrule \multirow{3}[2]{*}{Auto (392, 6)} & RMSE & \textbf{2.628} & \textbf{2.628} & \textbf{2.628} & 2.643 & 2.63 \\ 
   & \# predictors & \textbf{3} & 3.003 & \textbf{3} & 5.008 & 3.008 \\ 
   & running time (s) & 0.008 & 0.051 & 0.067 & 0.007 & 0.25 \\ 
  \midrule\multirow{3}[2]{*}{College (777, 17)} & RMSE & 1565.476 & 1568.234 & 1569.625 & \textbf{1564.807} & 1573.975 \\ 
   & \# predictors & 17.991 & 16.333 & 16.067 & 16.008 & \textbf{15.385} \\ 
   & running time (s) & 0.058 & 0.092 & 0.451 & 0.01 & 0.734 \\ 
  \midrule\multirow{3}[2]{*}{Forest Fires (517, 55)} & RMSE & \textbf{18.603} & 18.707 & 18.757 & 18.726 & 19.163 \\ 
   & \# predictors & \textbf{0} & 0.983 & 0.986 & 2.985 & 6.899 \\ 
   & running time (s) & 0.084 & 356.651 & 0.593 & 0.014 & 0.785 \\ 
   \bottomrule 
\end{tabular}
}
\end{table}

\section{Justification of hdf and AICc}
\label{sec:justification_aicc_hdf}

In this section, we provide justifications for hdf and AICc as a heuristic degrees of freedom and a selection rule for BS, respectively, when the predictors $X$ are orthogonal.

\subsection{The heuristic degrees of freedom}
\label{sec:justification_hdf}

We provide a theoretical justification for hdf in a restricted scenario where $X$ is orthogonal and the true model is null. We further give numerical justification for a general true model, and demonstrate the use of hdf in model selection for BS. 

\subsubsection{Theoretical justification of hdf under a null true model}
Assume $\mu=0$, with $X$ still being orthogonal. In such a restricted scenario, df$_C(k)$ (edf of BS for subset size $k$) has an analytical expression, which allows us to provide some theoretical justification for hdf$(k)$. We start by introducing notation, and present the main result in Theorem \ref{thm:hdf_ydf_representation} and its Corollary. The detailed proofs are given in the Supplemental Material Section \ref{sec:proof_hdf_ydf}.

Denote $\tilde{X}_{(i)}$ as the $i$-th largest order statistic in an i.i.d sample of size $p$ from a $\chi^2_1$ distribution. \citet{Ye1998} showed that
\begin{equation*}
\text{df}_C(k) = E\left( \sum_{i=1}^{k} \tilde{X}_{(i)} \right).
\end{equation*}
Let $\tilde{H}(s) = -\tilde{Q}(1-s)$ where $\tilde{Q}$ is the quantile function of a $\chi_1^2$ distribution, and $s\in (0,1)$. For $0\le s \le t \le 1$, the truncated variance function is defined as
\begin{equation*}
\tilde{\sigma}^2(s,t) = \int_{s}^{t} \int_{s}^{t} (u \wedge v -uv) d \tilde{H}(u) d \tilde{H}(v),
\end{equation*}
where $u \wedge v =\min(u,v)$. Denote $\tilde{Y}_p = \tilde{\sigma}_p^{-1}(\sum_{i=1}^k \tilde{X}_{(i)} - \tilde{\mu}_p)$, where
\begin{equation*}
\tilde{\sigma}_p = \sqrt{p} \cdot \tilde{\sigma}(1/p,k/p),
\end{equation*}
and
\begin{equation*}
\tilde{\mu}_p = -p \int_{1/p}^{k/p} \tilde{H}(u) du - \tilde{H}\left(\frac{1}{p}\right).
\end{equation*}

\begin{restatable}{theorem}{dfasy}
	\label{thm:hdf_ydf_representation}
	Assume $X$ is orthogonal and the true model is null ($\mu=0$). As $p\rightarrow  \infty$,  $k\rightarrow  \infty$ with $k=\left \lfloor{px}\right \rfloor$, we have
	\begin{equation}
	\label{eq:hdf_ydf_yp_representation}
	\frac{1}{2p} \text{hdf}(k) = \frac{1}{2p}\text{df}_C(k) - \frac{\tilde{\sigma}_p}{2p}E(\tilde{Y}_p) + O\left(\frac{\log(p)}{p} \right),
	\end{equation}
	where $x \in (0,1)$ is a constant and $\left \lfloor{\cdot}\right \rfloor$ denotes the greatest integer function.
\end{restatable}

\begin{restatable}{corollary}{dfasycorollary}
	\label{thm:hdf_ydf_corollary}
	If $\limsup |E(\tilde{Y}_p)| < \infty$ , we further have
	\begin{equation}
	\label{eq:main_corollary}
	\frac{\text{df}_C(k)}{\text{hdf}(k)} \rightarrow 1.
	\end{equation}
\end{restatable}
\noindent\textbf{Remark:} If $\tilde{Y}_p$ is uniformly integrable, then $E(\tilde{Y}_p) \rightarrow 0$, and hence the result of Corollary \ref{thm:hdf_ydf_corollary} holds.

It can be seen that Corollary \ref{thm:hdf_ydf_corollary} holds given the assumptions, since both hdf$(k)$ and df$_C(k)$ diverge while $E(\tilde{Y}_p)$ and the remainder term remain bounded. The Corollary suggests that for large $k$ and large $p$, the ratio of df$_C(k)$ to hdf$(k)$ will be close to $1$. We next explore empirically the relative behavior of the two dfs for a fixed $p$ with an increasing $k$. 

\subsubsection{Numerical justification of hdf}
Figure \ref{fig:dfc_dflambda} shows the comparison of hdf and edf via simulations. We fit BS on $1000$ realizations of the response generated after fixing $X$. The edf is calculated based on definition \eqref{eq:edf} using the sample covariances, while hdf is given by Algorithm \ref{alg:hdf}. We see that in the null case, using hdf to approximate edf becomes more accurate as $k$ approaches $p$, providing a finite-sample justification of Corollary \ref{thm:hdf_ydf_corollary}. 

In addition to the null model, we consider a sparse model (Orth-Sparse-Ex1) with $p_0=6$ true predictors (those with non-zero coefficients), and a dense model (Orth-Dense) where all predictors have non-zero coefficients. Similarly to the null case, we see that hdf approaches edf as $k$ gets close to $p$, i.e. the statement of Corollary \ref{thm:hdf_ydf_corollary} holds in these scenarios as well. Furthermore, we see that hdf generally approximates edf well, where the difference is more pronounced when BS underfits, e.g. a sparse model with high SNR and $k<p_0=6$ or a dense model with high SNR with $k<p=14$. Clearly, underfitting causes the problem, particularly when what is left out is important, such as in a high SNR case.

\begin{figure}[!ht]
	\centering
	\includegraphics[width=0.9\textwidth]{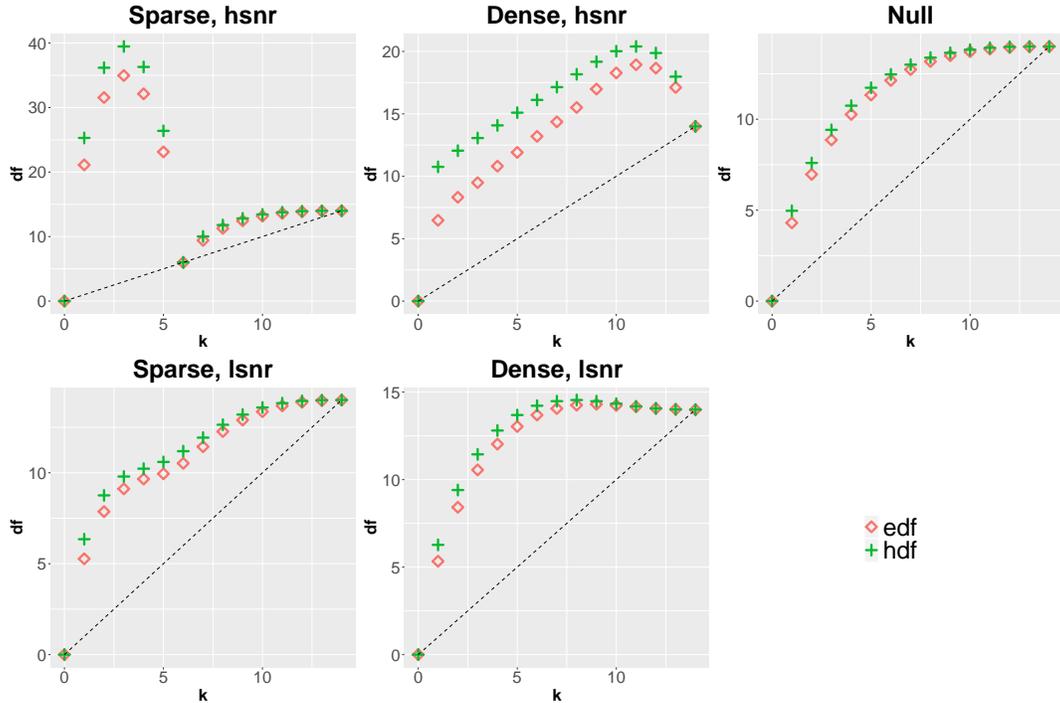}
	\caption{hdf$(k)$ and df$_C(k)$ (edf of constrained BS). The black dashed line is the 45-degree line. Here $X$ is orthogonal with $n=200$ and $p=14$. Three types of the true model and two SNR are considered. We assume knowledge of $\mu$ and $\sigma$.}
	\label{fig:dfc_dflambda}
\end{figure}

\subsubsection{\texorpdfstring{C$_p$}{Lg}-hdf as a feasible implementation of \texorpdfstring{C$_p$}{Lg}-edf }
\label{sec:cp_edf_hdf}
We have shown that hdf generally approximates edf well, and agrees with edf in some situations as $k$ approaches $p$. By replacing edf with hdf in \eqref{eq:cp_edf}, we have a feasible selection rule C$_p$-hdf. Figure \ref{fig:cp_edf_hdf} compares the averages of C$_p$-edf and C$_p$-hdf over $1000$ replications. Similarly to the comparison of the degrees of freedom values, we see C$_p$-hdf agrees with C$_p$-edf on average as $k$ approaches $p$. Even at the places where we see differences between the degrees of freedom values, e.g. a sparse true model with high SNR and $k<p_0=6$, the differences are compensated by the model fit and we see C$_p$-hdf is very close to C$_p$-edf. As we discussed in Section \ref{sec:optimism}, for any general fitting procedure including BS, C$_p$-edf provides an unbiased estimator of the expected prediction error where the error measure $\Theta$ is the squared error (SE), i.e. $E(\text{C}_p\text{-edf}) = E(\text{Err}_{\text{SE}})$. Therefore, by using the sample average to represent the population mean, Figure \ref{fig:cp_edf_hdf} indicates that $E(\text{C}_p\text{-hdf})$ approximates $E(\text{Err}_{\text{SE}})$ well, and moreover C$_p$-hdf gives the same average selected size as C$_p$-edf in all cases, when they are applied as selection rules, supporting the use of hdf in model selection for BS under orthogonal predictors.

\begin{figure}[!ht]
	\centering
	\includegraphics[width=0.9\textwidth]{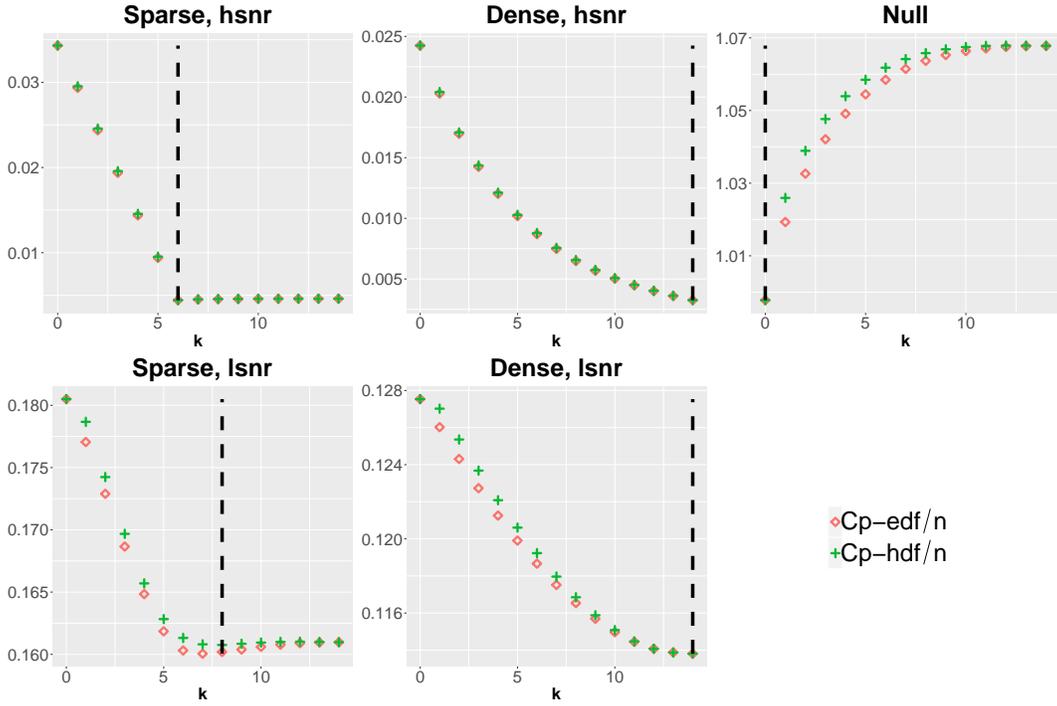}
	\caption{Averages of C$_p$-edf and C$_p$-hdf over $1000$ replications. Both criteria lead to the same average of the selected subset size over the $1000$ replications (rounded to the nearest integer), as denoted by the black dashed vertical lines. Other details are the same as in Figure \ref{fig:dfc_dflambda}.}
	\label{fig:cp_edf_hdf} 
\end{figure}

\subsection{A KL-based information criterion for BS}
When the error measure $\Theta$ is the deviance \eqref{eq:deviance_def}, the prediction error $\text{Err}_\text{KL}$ is the KL discrepancy. AICc-edf \eqref{eq:aicc_edf} is motivated by trying to construct an unbiased estimator of $E(\text{Err}_\text{KL})$. The expected KL-based optimism for BS is given as
\begin{equation}
E(\text{op}_\text{KL}) = E\left(n \frac{n\sigma^2+\lVert \mu-X\hat{\beta}(k) \rVert_2^2}{\lVert y-X\hat{\beta}(k)\rVert_2^2}\right) + n.
\label{eq:eop_expression}
\end{equation}
Note that \eqref{eq:eop_expression} holds for a general $X$. Augmenting $E(\text{op}_\text{KL})$ with the training error $\text{err}_{\text{KL}}$ we have $\widehat{\text{Err}}_{\text{KL}}$ according to \eqref{eq:err_eop}, where (ignoring the constant $n\log(2\pi)$ for convenience)
\begin{equation}
\text{err}_{\text{KL}} = n\log\left(\frac{\text{RSS}}{n}\right) - n,
\label{eq:err_expression}
\end{equation}
since the pre-specified model $f$ in \eqref{eq:deviance_def} follows a Gaussian distribution as assumed in \eqref{eq:truemodel_def}. The derivations of \eqref{eq:eop_expression} and \eqref{eq:err_expression} are presented in the Supplemental Material Section \ref{sec:expectedkl_bs}. 

Figure \ref{fig:eop_approx_withrss} shows the averages of AICc-edf, AICc-hdf and $\widehat{\text{Err}}_{\text{KL}}$ over $1000$ replications. By using the sample average to represent the population mean, we first see that $E(\text{AICc-edf})$ generally tracks the expected KL, $E(\text{Err}_{\text{KL}})$ reasonably well. In fact, they agree with each other in the null case and a sparse true model with high SNR. Noticeable discrepancies can be observed in a sparse true model with high SNR and $k<p_0=6$. This is the place where the set of true predictors is not entirely included in the model. The derivations of the classic AIC and AICc (both with ndf plugged in according to our notation) are based on an assumption that the true predictors are included in the model. In the situation where this assumption is violated, AICc will no longer be unbiased, and a similar conjecture can be made here for AICc-edf in the context of BS. Second, similarly to the comparison of $E(\text{C}_p\text{-edf})$ and $E(\text{C}_p\text{-hdf})$ in Section \ref{sec:cp_edf_hdf}, we see that $E(\text{AICc-hdf})$ approximates $E(\text{AICc-edf})$ well and they agree with each other as $k$ approaches $p$. 
Last and most importantly, both AICc-edf and AICc-hdf yield the same average selected size as $\widehat{\text{Err}}_{\text{KL}}$ across all scenarios, supporting the use of AICc-hdf as a selection rule for BS under orthogonal predictors.

\begin{figure}[!ht]
	\centering
	\includegraphics[width=0.9\textwidth]{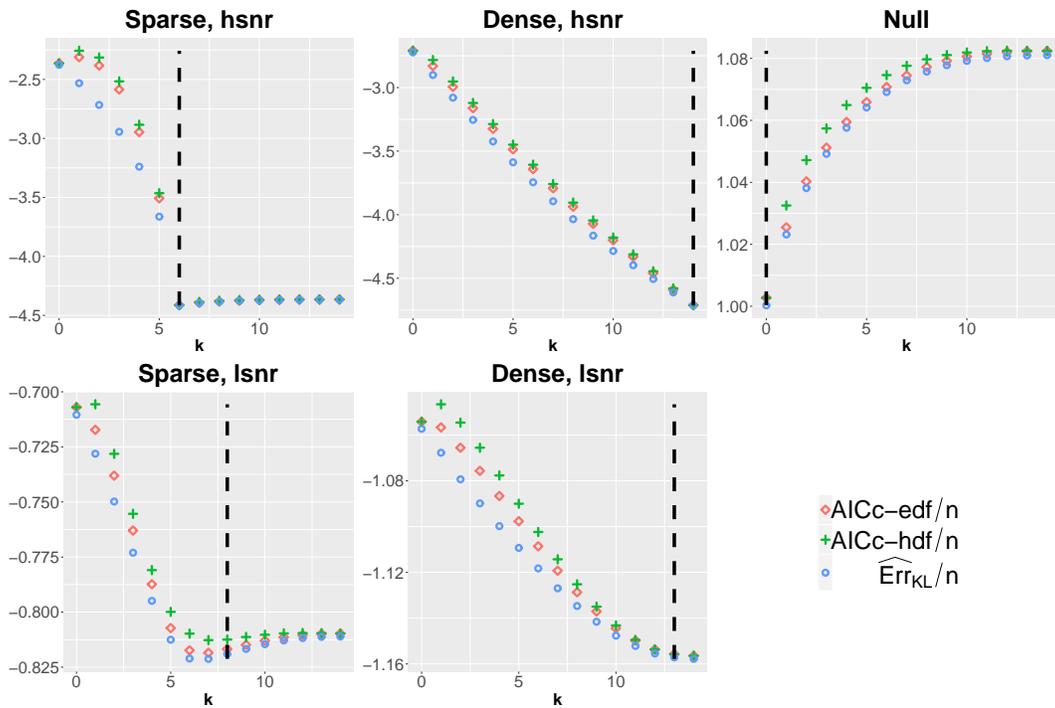}
	\caption{Averages of AICc-edf, AICc-hdf and $\widehat{\text{Err}}_\text{KL}$ over $1000$ replications. All three criteria lead to the same average of the selected subset size over the $1000$ replications (rounded to the nearest integer), as denoted by the black dashed vertical lines. Other details are the same as in Figure \ref{fig:dfc_dflambda}.}
	\label{fig:eop_approx_withrss} 
\end{figure}

\subsection{A discussion on the use of information criteria in LBS}
Since for orthogonal $X$ the edf of LBS has an analytical expression and LBS can recover the solution path of BS, one may ask why not just use LBS with a selection rule such as C$_p$-edf, which is well-defined for any general fitting procedure. 

We consider a fixed sequence of $\lambda$ and compute the LBS solutions for $1000$ realizations. The decreasing sequence of $\lambda$ starts at the smallest value $\lambda_{\text{max}}$ for which the estimated coefficient vector $\hat{\beta}$ equals zero for all of the $1000$ realizations. We then construct a sequence of $200$ values of $\lambda$ equally spaced in log scale from $\lambda_{\text{max}}$ to $\lambda_{\text{min}}$, where $\lambda_{\text{min}} = \alpha \lambda_{\text{max}}$ and $\alpha=0.001$. This procedure of generating the sequence of $\lambda$ has been discussed by \citet{Friedman2010} in the context of lasso.

We find that (see Supplemental Material Table \ref{tab:lbs_bs}) LBS is outperformed by BS in almost all scenarios based on $1000$ simulation replications. We use C$_p$-edf as the selection rule for both methods, where edf of BS (df$_C(k)$) is estimated via simulations and edf of LBS (df$_L(\lambda)$) is calculated using formulas \eqref{eq:thdf_expression} and \eqref{eq:thdf_size_expression}. We see that 1) LBS deteriorates as $p$ gets larger when SNR is low or sample size $n$ is large; and 2) increasing the sample size $n$ does not help LBS. We further compare the number of predictors given by BS and LBS for each of the $1000$ replications (see Supplemental Material Figure \ref{fig:numvar_bs_lbs}), where we consider the Orth-Sparse-Ex1 model with $n=200$ and high SNR. Under this design, LBS never selects fewer predictors than BS, and it chooses more predictors in $31.1\%$ and $37.8\%$ of all replications for $p=30$ and $p=180$, respectively. A possible explanation for this might be that df$_L(\lambda)$ characterizes the model complexity at $\lambda$ on average, but does not correctly describe the model complexity on a given realization. Given a single realization, there are an infinite number of $\lambda$ values that correspond to each distinct solution, and they lead to different values of df$_L(\lambda)$ and further result in different model complexities and different C$_p$ values. This variability in the C$_p$ values for the same solution potentially causes the selected subsets of LBS to have more variabilities than those selected subsets of BS.

\section{Conclusion and future work}
\label{sec:conclusion}
In this paper, we propose an LS-based subset selection method BOSS. In order to select the single optimal subset, we introduce a heuristic degrees of freedom (hdf) for BOSS and a KL-based information criterion AICc-hdf. BOSS together with AICc-hdf has computational cost on the same order as a single LS fit. We find that AICc-hdf is the best selection rule compared to other information criteria and CV. We further show in simulations and real data examples that BOSS using AICc-hdf is competitive in both speed and predictive performance. Finally, we provide justifications of hdf and AICc in a restricted scenario where $X$ is orthogonal.

The strong performance of BOSS using AICc compared to using CV suggests that the pursuit of methods to approximate edf (which normally does not have an analytical expression for complex modeling methods and algorithms), particularly for methods that are more sensitive to small perturbations in the data, is worthy of further research. Also, we focus on a fixed $X$ in this paper, however, in many applications where the data are observational and the experiment is performed in an uncontrolled manner, it is more appropriate to treat $X$ as random. It is interesting to study the performance of BOSS using an information criterion designed for random predictors, e.g. RAICc \citep{tian2020selection}.

\clearpage
\bibliographystyle{chicago}
\bibliography{reference.bib}

\begin{thebibliography}{}

\bibitem[\protect\citeauthoryear{Akaike}{Akaike}{1973}]{Akaike1973}
Akaike, H. (1973).
\newblock Information theory and an extension of the maximum likelihood
  principle.
\newblock In B.~P.~F. Csaki (Ed.), {\em 2nd International Symposium on
  Information Theory}, Budapest, Hungary, pp.\  267--281. Akadémiai Kiadó.

\bibitem[\protect\citeauthoryear{Bertsimas, King, and Mazumder}{Bertsimas
  et~al.}{2016}]{Bertsimas2016}
Bertsimas, D., A.~King, and R.~Mazumder (2016).
\newblock Best subset selection via a modern optimization lens.
\newblock {\em The Annals of Statistics\/}~{\em 44\/}(2), 813--852.

\bibitem[\protect\citeauthoryear{Bertsimas, Pauphilet, Van~Parys,
  et~al.}{Bertsimas et~al.}{2020}]{bertsimas2020sparse}
Bertsimas, D., J.~Pauphilet, B.~Van~Parys, et~al. (2020).
\newblock Sparse regression: Scalable algorithms and empirical performance.
\newblock {\em Statistical Science\/}~{\em 35\/}(4), 555--578.

\bibitem[\protect\citeauthoryear{Cortez and Morais}{Cortez and
  Morais}{2007}]{cortez2007data}
Cortez, P. and A.~Morais (2007).
\newblock A data mining approach to predict forest fires using meteorological
  data. {In J. Neves, M. F. Santos and J. Machado Eds.}
\newblock {\em New Trends in Artificial Intelligence, Proceedings of the 13th
  EPIA 2007 - Portuguese Conference on Artificial Intelligence\/}, 512--523.

\bibitem[\protect\citeauthoryear{Efron}{Efron}{1986}]{Efron1986}
Efron, B. (1986).
\newblock How biased is the apparent error rate of a prediction rule?
\newblock {\em Journal of the American Statistical Association\/}~{\em
  81\/}(394), 461--470.

\bibitem[\protect\citeauthoryear{Efron}{Efron}{2004}]{Efron2004}
Efron, B. (2004).
\newblock The estimation of prediction error: covariance penalties and
  cross-validation.
\newblock {\em Journal of the American Statistical Association\/}~{\em
  99\/}(467), 619--632.

\bibitem[\protect\citeauthoryear{Efron, Hastie, Johnstone, Tibshirani,
  et~al.}{Efron et~al.}{2004}]{efron2004least}
Efron, B., T.~Hastie, I.~Johnstone, R.~Tibshirani, et~al. (2004).
\newblock Least angle regression.
\newblock {\em Annals of statistics\/}~{\em 32\/}(2), 407--499.

\bibitem[\protect\citeauthoryear{Flynn, Hurvich, and Simonoff}{Flynn
  et~al.}{2013}]{Flynn2013}
Flynn, C.~J., C.~M. Hurvich, and J.~S. Simonoff (2013).
\newblock Efficiency for regularization parameter selection in penalized
  likelihood estimation of misspecified models.
\newblock {\em Journal of the American Statistical Association\/}~{\em
  108\/}(503), 1031--1043.

\bibitem[\protect\citeauthoryear{Friedman, Hastie, and Tibshirani}{Friedman
  et~al.}{2010}]{Friedman2010}
Friedman, J., T.~Hastie, and R.~Tibshirani (2010).
\newblock Regularization paths for generalized linear models via coordinate
  descent.
\newblock {\em Journal of Statistical Software\/}~{\em 33\/}(1), 1.

\bibitem[\protect\citeauthoryear{Furnival and Wilson}{Furnival and
  Wilson}{1974}]{Furnival1974}
Furnival, G.~M. and R.~W. Wilson (1974).
\newblock Regressions by leaps and bounds.
\newblock {\em Technometrics\/}~{\em 16\/}(4), 499--511.

\bibitem[\protect\citeauthoryear{Hammarling and Lucas}{Hammarling and
  Lucas}{2008}]{hammarling2008updating}
Hammarling, S. and C.~Lucas (2008).
\newblock Updating the qr factorization and the least squares problem.
\newblock MIMS EPrint: 2008.111, reports available from
  http://eprints.maths.manchester.ac.uk/.

\bibitem[\protect\citeauthoryear{Harris}{Harris}{2016}]{Harris2016}
Harris, X.~T. (2016).
\newblock Prediction error after model search.
\newblock {\em arXiv preprint arXiv:1610.06107\/}.

\bibitem[\protect\citeauthoryear{Hastie, Tibshirani, and Tibshirani}{Hastie
  et~al.}{2017}]{Hastie2017}
Hastie, T., R.~Tibshirani, and R.~J. Tibshirani (2017).
\newblock Extended comparisons of best subset selection, forward stepwise
  selection, and the lasso.
\newblock {\em arXiv preprint arXiv:1707.08692\/}.

\bibitem[\protect\citeauthoryear{Hazimeh and Mazumder}{Hazimeh and
  Mazumder}{2020}]{hazimeh2020fast}
Hazimeh, H. and R.~Mazumder (2020).
\newblock Fast best subset selection: Coordinate descent and local
  combinatorial optimization algorithms.
\newblock {\em Operations Research\/}~{\em 68\/}(5), 1517--1537.

\bibitem[\protect\citeauthoryear{Hocking and Leslie}{Hocking and
  Leslie}{1967}]{Hocking1967}
Hocking, R. and R.~Leslie (1967).
\newblock Selection of the best subset in regression analysis.
\newblock {\em Technometrics\/}~{\em 9\/}(4), 531--540.

\bibitem[\protect\citeauthoryear{Hurvich and Tsai}{Hurvich and
  Tsai}{1989}]{Hurvich1989}
Hurvich, C.~M. and C.-L. Tsai (1989).
\newblock Regression and time series model selection in small samples.
\newblock {\em Biometrika\/}~{\em 76\/}(2), 297--307.

\bibitem[\protect\citeauthoryear{Hurvich and Tsai}{Hurvich and
  Tsai}{1991}]{Hurvich1991}
Hurvich, C.~M. and C.-L. Tsai (1991).
\newblock Bias of the corrected {AIC} criterion for underfitted regression and
  time series models.
\newblock {\em Biometrika\/}~{\em 78\/}(3), 499--509.

\bibitem[\protect\citeauthoryear{Ing and Lai}{Ing and
  Lai}{2011}]{ing2011stepwise}
Ing, C.-K. and T.~L. Lai (2011).
\newblock A stepwise regression method and consistent model selection for
  high-dimensional sparse linear models.
\newblock {\em Statistica Sinica\/}, 1473--1513.

\bibitem[\protect\citeauthoryear{Janson, Fithian, and Hastie}{Janson
  et~al.}{2015}]{janson2015effective}
Janson, L., W.~Fithian, and T.~J. Hastie (2015).
\newblock Effective degrees of freedom: a flawed metaphor.
\newblock {\em Biometrika\/}~{\em 102\/}(2), 479--485.

\bibitem[\protect\citeauthoryear{Konishi and Kitagawa}{Konishi and
  Kitagawa}{2008}]{konishi2008information}
Konishi, S. and G.~Kitagawa (2008).
\newblock {\em Information Criteria and Statistical Modeling}.
\newblock Berlin: Springer Science \& Business Media.

\bibitem[\protect\citeauthoryear{Liao, Cavanaugh, and McMurry}{Liao
  et~al.}{2018}]{Liao2018}
Liao, J., J.~E. Cavanaugh, and T.~L. McMurry (2018).
\newblock Extending {AIC} to best subset regression.
\newblock {\em Computational Statistics\/}~{\em 33\/}(2), 787--806.

\bibitem[\protect\citeauthoryear{Lumley}{Lumley}{2017}]{FortrancodebyAlanMiller2020}
Lumley, T. (2017).
\newblock {\em leaps: Regression Subset Selection, based on Fortran code by
  Alan Miller}.
\newblock R package version 3.1.

\bibitem[\protect\citeauthoryear{Mallows}{Mallows}{1973}]{mallows1973some}
Mallows, C.~L. (1973).
\newblock Some comments on {Cp}.
\newblock {\em Technometrics\/}~{\em 15\/}(4), 661--675.

\bibitem[\protect\citeauthoryear{Mazumder, Friedman, and Hastie}{Mazumder
  et~al.}{2011}]{Mazumder2011}
Mazumder, R., J.~H. Friedman, and T.~Hastie (2011).
\newblock Sparsenet: Coordinate descent with nonconvex penalties.
\newblock {\em Journal of the American Statistical Association\/}~{\em
  106\/}(495), 1125--1138.

\bibitem[\protect\citeauthoryear{Meinshausen}{Meinshausen}{2007}]{Meinshausen2007}
Meinshausen, N. (2007).
\newblock Relaxed lasso.
\newblock {\em Computational Statistics \& Data Analysis\/}~{\em 52\/}(1),
  374--393.

\bibitem[\protect\citeauthoryear{Natarajan}{Natarajan}{1995}]{Natarajan1995}
Natarajan, B.~K. (1995).
\newblock Sparse approximate solutions to linear systems.
\newblock {\em SIAM Journal on Computing\/}~{\em 24\/}(2), 227--234.

\bibitem[\protect\citeauthoryear{Reid, Tibshirani, and Friedman}{Reid
  et~al.}{2016}]{reid2016study}
Reid, S., R.~Tibshirani, and J.~Friedman (2016).
\newblock A study of error variance estimation in lasso regression.
\newblock {\em Statistica Sinica\/}, 35--67.

\bibitem[\protect\citeauthoryear{Schwarz}{Schwarz}{1978}]{schwarz1978estimating}
Schwarz, G. (1978).
\newblock Estimating the dimension of a model.
\newblock {\em The Annals of Statistics\/}~{\em 6\/}(2), 461--464.

\bibitem[\protect\citeauthoryear{Taddy}{Taddy}{2017}]{Taddy2017}
Taddy, M. (2017).
\newblock One-step estimator paths for concave regularization.
\newblock {\em Journal of Computational and Graphical Statistics\/}~{\em
  26\/}(3), 525--536.

\bibitem[\protect\citeauthoryear{Tian, Hurvich, and Simonoff}{Tian
  et~al.}{2020}]{tian2020selection}
Tian, S., C.~M. Hurvich, and J.~S. Simonoff (2020).
\newblock Selection of regression models under linear restrictions for fixed
  and random designs.
\newblock {\em arXiv preprint arXiv:2009.10029\/}.

\bibitem[\protect\citeauthoryear{Tibshirani}{Tibshirani}{1996}]{Tibshirani1996}
Tibshirani, R. (1996).
\newblock Regression shrinkage and selection via the lasso.
\newblock {\em Journal of the Royal Statistical Society. Series B
  (Methodological)\/}~{\em 58\/}(1), 267--288.

\bibitem[\protect\citeauthoryear{Tibshirani}{Tibshirani}{2015}]{Tibshirani2015}
Tibshirani, R.~J. (2015).
\newblock Degrees of freedom and model search.
\newblock {\em Statistica Sinica\/}, 1265--1296.

\bibitem[\protect\citeauthoryear{Wang}{Wang}{2009}]{wang2009forward}
Wang, H. (2009).
\newblock Forward regression for ultra-high dimensional variable screening.
\newblock {\em Journal of the American Statistical Association\/}~{\em
  104\/}(488), 1512--1524.

\bibitem[\protect\citeauthoryear{Ye}{Ye}{1998}]{Ye1998}
Ye, J. (1998).
\newblock On measuring and correcting the effects of data mining and model
  selection.
\newblock {\em Journal of the American Statistical Association\/}~{\em
  93\/}(441), 120--131.

\end{thebibliography}


\begin{thebibliography}{}

\bibitem[\protect\citeauthoryear{Cs\"{o}rg\H{o}, Haeusler, and
  Mason}{Cs\"{o}rg\H{o} et~al.}{1991}]{Csorgo1991}
Cs\"{o}rg\H{o}, S., E.~Haeusler, and D.~M. Mason (1991).
\newblock The asymptotic distribution of extreme sums.
\newblock {\em The Annals of Probability\/}~{\em 19(2)}, 783--811.

\bibitem[\protect\citeauthoryear{Embrechts, Kl{\"u}ppelberg, and
  Mikosch}{Embrechts et~al.}{2013}]{Embrechts2013}
Embrechts, P., C.~Kl{\"u}ppelberg, and T.~Mikosch (2013).
\newblock {\em Modelling Extremal Events: for Insurance and Finance},
  Volume~33.
\newblock Berlin: Springer Science \& Business Media.

\bibitem[\protect\citeauthoryear{Fung and Seneta}{Fung and
  Seneta}{2017}]{Fung2017}
Fung, T. and E.~Seneta (2017).
\newblock Quantile function expansion using regularly varying functions.
\newblock {\em Methodology and Computing in Applied Probability\/}~{\em 20(4)},
  1091--1103.

\end{thebibliography}

\beginsupplement
\appendix
\pagenumbering{arabic}
\begin{center}
\textbf{\large Supplementary Materials \\
On the Use of Information Criteria for Subset Selection \\ in Least Squares Regression}

Sen Tian, Clifford M. Hurvich, Jeffrey S. Simonoff
\end{center}

\section{Technical details}
\subsection{Proof of theorem \ref{thm:hdf_ydf_representation} and its corollary}
\label{sec:proof_hdf_ydf}
In this section, we assume an orthogonal $X$ and a null true model. This is the only scenario under which both df$_C(k)$ and hdf$(k)$ have analytical expressions. We will prove that the ratio of df$_C(k)$ and hdf$(k)$ goes to $1$ as $k,p\rightarrow \infty$ while $k=\left \lfloor{xp}\right \rfloor $, where $\left \lfloor{\cdot}\right \rfloor$ denotes the greatest integer function and $x\in(0,1)$. We start by laying out a few lemmas to be used in the proof of the main theorem.
\begin{lemma}
	\label{lemma:hdf_nulltrue}
	Assume the design matrix is orthogonal and the true model is null ($\mu=0$). Then
	\begin{equation}
	\text{hdf}(k) = df_L(\lambda_k^\star) = k - 2p\cdot \Phi^{-1} \left(\frac{k}{2p}\right) \cdot \phi\left[\Phi^{-1}\left(\frac{k}{2p}\right) \right].
	\label{eq:hdf_nulltrue}
	\end{equation}
\end{lemma}
\begin{proof}
	We follow the steps described in algorithm \ref{alg:hdf}. We first find $\lambda_k^\star$ from \eqref{eq:thdf_size_expression}, by using the fact that $\mu=0$, and we get $\displaystyle -\frac{\sqrt{2\lambda_k^\star}}{\sigma} = \displaystyle \Phi^{-1}\left(\frac{k}{2p}\right)$, which we then substituted into \eqref{eq:thdf_expression} to get \eqref{eq:hdf_nulltrue}.
\end{proof}

\begin{lemma}
	\label{lemma:G(x)}
	Define $\tilde{G}(x)=  x-\Phi^{-1}(x)\cdot \phi \left[\Phi^{-1}(x)\right]$, where $x\in (0,1)$ is a continuous variable. We have
	\begin{equation*}
	\lim_{x\to 0} \tilde{G}(x) = 0,
	\end{equation*}
	and 
	\begin{equation*}
	\tilde{G}^\prime(x) = \left[\Phi^{-1}(x)\right]^2.
	\end{equation*}
	Therefore by the fundamental theorem of calculus,
	\begin{equation*}
	\tilde{G}(x) = \int_{0}^{x} \left[\Phi^{-1}(u)\right]^2 du.
	\end{equation*}
\end{lemma}
\begin{proof}
	First note that, since $\phi^\prime(v)= -v\cdot \phi(v)$ and $\lim_{v\to \pm \infty} \phi^\prime(v) =0$, we have
	\begin{equation*}
	\lim_{v\to \pm \infty} v \cdot \phi(v) = 0.
	\end{equation*}
	Let $v=\Phi^{-1}(x)$. Then
	\begin{equation*}
	\lim_{x\to 0} \tilde{G}(x)  = \lim_{v\to -\infty} -v \cdot \phi(v) = 0.
	\end{equation*}

	Next, we obtain the derivative of $\tilde{G}(x)$. Since $\Phi^\prime(x) = \phi(x)$, we have
	\begin{equation}
	\left[\Phi^{-1}(x)\right]^\prime = \frac{1}{\Phi^\prime \left[\Phi^{-1}(x)\right]}=\frac{1}{\phi \left[\Phi^{-1}(x)\right]}.
	\label{eq:G(x)_derivative_1}
	\end{equation}
	Also since $\phi^\prime(x) = -x\cdot \phi(x)$, we have
	\begin{equation}
	\phi^\prime\left[\Phi^{-1}(x)\right] = - \Phi^{-1}(x) \cdot \phi\left[\Phi^{-1}(x)\right] \cdot \left[\Phi^{-1}(x)\right]^\prime = -\Phi^{-1}(x).
	\label{eq:G(x)_derivative_2}
	\end{equation}
	By \eqref{eq:G(x)_derivative_1} and \eqref{eq:G(x)_derivative_2}, we have
	\begin{equation*}
	\tilde{G}^\prime(x) = 1 - \left[\Phi^{-1}(x)\right]^\prime \cdot  \phi\left[\Phi^{-1}(x)\right] - \left[\Phi^{-1}(x)\right] \cdot \phi^\prime\left[\Phi^{-1}(x)\right] =  \left[\Phi^{-1}(x)\right]^2.
	\end{equation*}
	
	Therefore, by the fundamental theorem of calculus, we have
	\begin{equation*}
	\tilde{G}(x) = \int_{0}^{x} \tilde{G}^\prime(u) du + \tilde{G}(0) = \int_{0}^{x} \left[\Phi^{-1}(u)\right]^2 du.
	\end{equation*}
\end{proof}

\begin{lemma}
	\label{lemma:sigmasq}
	Denote $\tilde{Q}$ as the quantile function of a $\chi_1^2$ distribution, and let $\tilde{H}(s) = -\tilde{Q}(1-s)$ where $s\in (0,1)$. For $0\le s \le t \le 1$, consider the truncated variance function
	\begin{equation}
	\tilde{\sigma}^2(s,t) = \int_{s}^{t} \int_{s}^{t} (u \wedge v -uv) d \tilde{H}(u) d \tilde{H}(v),
	\label{eq:sigmasq}
	\end{equation}
	where $u \wedge v =\min(u,v)$. We have
	\begin{equation*}
	0 \le \tilde{\sigma}^2(s,t) \le 1.
	\end{equation*}
\end{lemma}
\begin{proof}
	We first note three facts.
	\begin{align}
	\tilde{H}(s) &= -\left[\Phi^{-1}\left(1-\frac{s}{2}\right) \right]^2=-\left[\Phi^{-1}\left(\frac{s}{2}\right) \right]^2,\label{eq:Hs} \\
	d\tilde{H}(s) &= \frac{\Phi^{-1}(1-s/2)}{\phi\left[\Phi^{-1}(1-s/2)\right]}ds=-\frac{\Phi^{-1}(s/2)}{\phi\left[\Phi^{-1}(s/2)\right]}ds, \quad \text{by \eqref{eq:G(x)_derivative_1}},\label{eq:dH} \\
	\Phi^{-1}(w) &= -\sqrt{\log\frac{1}{w^2} - \log\log\frac{1}{w^2} - \log(2\pi)} + o(1), \quad \text{for small $w$, by \citetonline{Fung2017}}. \label{eq:Phiinv_order}
	\end{align}
	Hence for small $w$,
	\begin{equation}
	\label{eq:Phiinvsq_order}
	\left[\Phi^{-1}(w)\right]^2 = O\left(\log \frac{1}{w^2}\right).
	\end{equation}
	Then by \eqref{eq:Hs} and \eqref{eq:Phiinvsq_order}, we have
	\begin{equation}
	\label{eq:sHs_limit}
	\lim_{s \to 0} s\cdot \tilde{H}(s) = \lim_{s \to 0} -s\cdot \left[\Phi^{-1}\left(\frac{s}{2}\right)\right]^2 = 0.
	\end{equation}
	Also, by \eqref{eq:Hs} and Lemma \ref{lemma:G(x)},
	\begin{equation}
	\label{eq:Hs_integral}
	-\int_{0}^{x} \tilde{H}(s) ds =  2\cdot \tilde{G}\left(\frac{x}{2}\right).
	\end{equation}
	Since $u,v\in[0,1]$, we have $u \wedge v-uv \ge 0$. By \eqref{eq:dH}, we have $d\tilde{H}(s)/ds \ge 0$. Therefore, the integrand in \eqref{eq:sigmasq} is non-negative, so that
	\begin{equation*}
	\tilde{\sigma}^2(s,t) \ge 0,
	\end{equation*}
	and 
	\begin{equation*}
	\begin{aligned}
	\tilde{\sigma}^2(s,t) &\le \int_{0}^{1} \int_{0}^{1} (u \wedge v -uv) d \tilde{H}(u) d \tilde{H}(v),\\
	&= \int_{0}^{1} \left[\int_{0}^{v} u(1-v) d\tilde{H}(u) + \int_{v}^{1} v(1-u) d\tilde{H}(u)   \right]  d\tilde{H}(v),\\
	&= \int_{0}^{1} \left[\int_{0}^{v} u d\tilde{H}(u) + v\int_{v}^{1} d\tilde{H}(u) -v\int_{0}^{1}u d\tilde{H}(u)  \right]  d\tilde{H}(v).
	\end{aligned}
	\end{equation*}
	Denote 
	\begin{equation*}
	\tilde{M}(v) = \int_{0}^{v} u d\tilde{H}(u) + v\int_{v}^{1} d\tilde{H}(u) -v\int_{0}^{1}u d\tilde{H}(u).
	\end{equation*}
	Now, we consider the three integrals in $\tilde{M}(v)$. First note that
	\begin{equation*}
	\begin{aligned}
	\int_{0}^{x} u d\tilde{H}(u) &= u\cdot \tilde{H}(u) \Bigr|_{0}^x - \int_{0}^{x} \tilde{H}(u) du,\\
	&= x\cdot \tilde{H}(x) - \int_{0}^{x} \tilde{H}(u) du,\quad \text{by \eqref{eq:sHs_limit}}\\
	&=  x\cdot \tilde{H}(x) + 2\cdot \tilde{G}(x/2), \quad \text{by \eqref{eq:Hs_integral}}.
	\end{aligned}
	\end{equation*}
	It is easily verified that $\tilde{H}(1)=0$ and $\tilde{G}(1/2)=1/2$, we have
	\begin{equation*}
	\int_{0}^{1} u d\tilde{H}(u) = 2\cdot \tilde{G}(1/2)=1,
	\end{equation*}
	and
	\begin{equation*}
	v\int_{v}^{1} d\tilde{H}(u) = -v\cdot \tilde{H}(v).
	\end{equation*}
	Therefore,
	\begin{equation*}
	\begin{aligned}
	\tilde{M}(v) &= v\cdot \tilde{H}(v) +2 \cdot \tilde{G}(v/2)-v\cdot \tilde{H}(v)-2v\cdot \tilde{G}(1/2)\\
	&=2 \cdot \tilde{G}(v/2) - v.
	\end{aligned}
	\end{equation*}
	Finally,
	\begin{equation*}
	\begin{aligned}
	\int_{0}^{1} \tilde{M}(v) d\tilde{H}(v) &= \int_{0}^{1} 2 \cdot \tilde{G}(v/2)d\tilde{H}(v) - \int_{0}^{1}v d\tilde{H}(v),\\
	&= -\int_{0}^{1} \Phi^{-1}\left(\frac{v}{2}\right)\cdot \phi\left[\Phi^{-1}\left(\frac{v}{2}\right)\right]d\tilde{H}(v),\quad \text{by the definition of $\tilde{G}(x)$},\\
	&= 2\int_{0}^{1/2} \left[\Phi^{-1}(v)\right]^2 dv, \quad \text{by \eqref{eq:dH}},\\
	&=2\cdot \tilde{G}(1/2),\\
	&=1.
	\end{aligned}
	\end{equation*}
	Therefore,
	\begin{equation*}
	0 \le \tilde{\sigma}^2(s,t) \le 1.
	\end{equation*}
\end{proof}

\begin{theorem}
	\label{thm:ydf_representation}
	Assume the design matrix is orthogonal and the true model is null ($\mu=0$). Let $\tilde{X}_{(i)}$ be the $i$-th largest order statistic in an i.i.d sample of size $p$ from a $\chi^2_1$ distribution. Denote $\tilde{Y}_p = \tilde{\sigma}_p^{-1}(\sum_{i=1}^k \tilde{X}_{(i)} - \tilde{\mu}_p)$, where
	\begin{equation*}
	\tilde{\sigma}_p = \sqrt{p} \cdot \sigma(1/p,k/p),
	\end{equation*}
	and
	\begin{equation*}
	\tilde{\mu}_p = -p \int_{1/p}^{k/p} \tilde{H}(u) du - \tilde{H}\left(\frac{1}{p}\right),
	\end{equation*}
	where $\sigma(s,t)$ and $\tilde{H}(x)$ are defined in Lemma \ref{lemma:sigmasq}.
	
	As $k \to \infty$, $p \to \infty$ and $k=\left \lfloor{px}\right \rfloor$ with $x \in (0,1)$, we have
	\begin{equation}
	\frac{\text{df}_C(k)}{2p} = \frac{1}{2p} E\left[ \sum_{i=1}^k \tilde{X}_{(i)} \right]=  \frac{\tilde{\sigma}_p}{2p}E(\tilde{Y}_p) + \tilde{G}\left(\frac{k}{2p}\right) + O\left(\frac{\log(p)}{p}\right),
	\label{eq:ydf/2p_representation}
	\end{equation}
	where $\left \lfloor{\cdot}\right \rfloor$ denotes the greatest integer function, $\tilde{G}(x)$ is defined in Lemma \ref{lemma:G(x)}.
	
\end{theorem}
\begin{proof}
	We first apply a result in \citetonline{Csorgo1991}, to show that $\tilde{Y}_p=\tilde{\sigma}_p^{-1}(\sum_{i=1}^{k} \tilde{X}_{(i)}-\tilde{\mu}_p)$ converges in distribution to a standard normal. We then show how $\tilde{\mu}_p$ can be expressed in terms of function G plus a remainder term, which further leads to expression \eqref{eq:ydf/2p_representation}. 
	
	It follows from \citetonline{Csorgo1991} Corollary 2, that if there exist centering and normalizing constants $c_p$ and $d_p>0$, s.t.
	\begin{equation}
	d_p^{-1}(\tilde{X}_{(1)} - c_p) \xrightarrow{D} Y, \quad \text{where Y is the standard Gumbel distribution},
	\label{eq:scorgo_condition_gumbel} 
	\end{equation}
	then as $k \to \infty$, $p \to \infty$ and $k=\left \lfloor{px}\right \rfloor$ with $x \in (0,1)$,
	\begin{equation}
	\left(\sum_{i=1}^k \tilde{X}_{(i)}  - \tilde{\mu}_p\right) / \tilde{\sigma}_p \xrightarrow{D} Z, \quad \text{where Z is standard normal}.
	\label{eq:scorgo_result}
	\end{equation}
	
	First, it follows from \citetonline{Embrechts2013} that \eqref{eq:scorgo_condition_gumbel} holds, with $c_p=2\log(p)-\log\log(p)-\log(\pi)$ and $d_p=2$.

	Next, we have
	
	\begin{equation*}
	\begin{aligned}
	\tilde{\mu}_p &= -p \int_{1/p}^{k/p} \tilde{H}(u) du - \tilde{H}\left(\frac{1}{p}\right),\\
	&= -p \int_{0}^{k/p} \tilde{H}(u) du + p \int_{0}^{1/p} \tilde{H}(u) du - \tilde{H}\left(\frac{1}{p}\right),\\
	&= 2p\cdot \tilde{G}\left(\frac{k}{2p}\right) - 2p \cdot \tilde{G}\left(\frac{1}{2p}\right) + \left[\Phi^{-1}\left(\frac{1}{2p}\right)\right]^2, \quad \text{by \eqref{eq:Hs_integral}}.
	\end{aligned}		
	\end{equation*}
	Also, since
	\begin{equation*}
	\begin{aligned}
	\tilde{G}(\frac{1}{2p}) &= \frac{1}{2p} - \Phi^{-1}\left(\frac{1}{2p}\right) \cdot \phi\left[\Phi^{-1}\left(\frac{1}{2p}\right) \right],\quad \text{by definition of $\tilde{G}(x)$ in Lemma \ref{lemma:G(x)}},\\
	&= \frac{1}{2p} - \frac{1}{\sqrt{2\pi}}\Phi^{-1}\left(\frac{1}{2p}\right) \cdot \exp\left(-\frac{1}{2}\left[\Phi^{-1}\left(\frac{1}{2p}\right) \right]^2\right),\\
	&= \frac{1}{2p} + \frac{1}{\sqrt{2\pi}} \cdot \left(\sqrt{\log(4p^2)-\log\log(4p^2)-\log(2\pi)}+o(1)\right)\cdot \\
	& \qquad \exp\left[-\frac{1}{2} \left(\log(4p^2)-\log\log(4p^2)-\log(2\pi) +o(1)\right)\right], \quad \text{by \eqref{eq:Phiinv_order}},\\
	&= \frac{1}{2p} + \left(\sqrt{\log(4p^2)-\log\log(4p^2)-\log(2\pi)}+o(1)\right)\cdot \frac{\sqrt{\log(4p^2)}}{2p},\\
	&= O\left(\frac{\log(p)}{p}\right).
	\end{aligned}
	\end{equation*}
	Also
	\begin{equation*}
	\begin{aligned}
	\frac{1}{2p}\left[\Phi^{-1}\left(\frac{1}{2p}\right)\right]^2 &= O\left( \frac{\log(p)}{p}\right), \quad \text{by \eqref{eq:Phiinv_order}},
	\end{aligned}
	\end{equation*}
	and hence
	\begin{equation*}
	\begin{aligned}
	\frac{\tilde{\mu}_p}{2p} &= \tilde{G}\left(\frac{k}{2p}\right) - \tilde{G}\left(\frac{1}{2p}\right) + \frac{1}{2p}\left[\Phi^{-1}\left(\frac{1}{2p}\right)\right]^2,\\
	&= \tilde{G}\left(\frac{k}{2p}\right
	) + O\left( \frac{\log(p)}{p}\right).
	\end{aligned}
	\end{equation*}
	Therefore, \eqref{eq:ydf/2p_representation} holds, i.e.
	\begin{equation*}
	\frac{\text{df}_C(k)}{2p}=\frac{1}{2p} E\left(\sum_{i=1}^{k} \tilde{X}_{(i)}\right) = \frac{\tilde{\sigma}_p}{2p} E(\tilde{Y}_p) + \frac{\tilde{\mu}_p}{2p}=\frac{\tilde{\sigma}_p}{2p} E(\tilde{Y}_p) + \tilde{G}\left(\frac{k}{2p}\right) + O\left( \frac{\log(p)}{p}\right).
	\end{equation*}
\end{proof}

\begin{corollary}
	\label{corollary:Yp_order}
	If $\limsup |E(\tilde{Y_p})| < \infty$, we further have:
	\begin{equation}
	\frac{\text{df}_C(k)}{2p} = \tilde{G}\left(\frac{k}{2p}\right) + O\left(\frac{\log(p)}{p}\right) + O\left(\frac{1}{\sqrt{p}}\right).
	\label{eq:ydf/2p_representation_remark}
	\end{equation}
\end{corollary}
\begin{proof}
	By Lemma \ref{lemma:sigmasq} we have $0 \le \sigma(1/p,k/p) \le 1$, and hence $\tilde{\sigma}_p = O(\sqrt{p})$. Therefore by Theorem \ref{thm:ydf_representation}, we have
	\begin{equation*}
	\frac{\text{df}_C(k)}{2p}=  \tilde{G}\left(\frac{k}{2p}\right) + O\left( \frac{\log(p)}{p}\right) + O\left(\frac{1}{\sqrt{p}}\right).
	\end{equation*}	
\end{proof}

\dfasy*

\begin{proof}
	By Lemma \ref{lemma:hdf_nulltrue}, we have
	\begin{equation*}
	\text{hdf}(k) = df_L(\tilde{M}^{-1}(k)) = k - 2p\cdot \Phi^{-1} \left(\frac{k}{2p}\right) \cdot \phi\left[\Phi^{-1}\left(\frac{k}{2p}\right) \right].
	\end{equation*}
	Then by the definition of $\tilde{G}(x)$ in Lemma \ref{lemma:G(x)},
	\begin{equation*}
	\frac{1}{2p} \text{hdf}(k) = \tilde{G}\left(\frac{k}{2p}\right).
	\end{equation*}
	By Theorem \ref{thm:ydf_representation}, we also have
	\begin{equation*}
	\frac{1}{2p}\text{df}_C(k) = \frac{\sigma_p}{2p}E(\tilde{Y}_p) + \tilde{G}\left(\frac{k}{2p}\right) + O\left(\frac{\log(p)}{p} \right).
	\end{equation*}
	Therefore, \eqref{eq:hdf_ydf_yp_representation} holds, i.e.
	\begin{equation*}
	\frac{1}{2p} \text{hdf}(k) = \frac{1}{2p}\text{df}_C(k) - \frac{\tilde{\sigma}_p}{2p}E(\tilde{Y}_p) + O\left(\frac{\log(p)}{p} \right).
	\end{equation*}
\end{proof}

\dfasycorollary*
\begin{proof}
	By Theorem \ref{thm:hdf_ydf_representation} and Corollary \ref{corollary:Yp_order},
	\begin{equation*}
	\frac{1}{2p} \text{hdf}(k) = \frac{1}{2p}\text{df}_C(k) + O\left(\frac{1}{\sqrt{p}}\right) + O\left(\frac{\log(p)}{p} \right).
	\end{equation*}
	From Lemma \ref{lemma:G(x)}, $\tilde{G}(x)$ is a non-decreasing function with $\tilde{G}(0+)=0$ and $\tilde{G}(1/2)=1/2$. Thus, 
	\begin{equation*}
	\frac{2p}{\text{hdf}(k)} = \frac{1}{\tilde{G}\left(\frac{k}{2p}\right)} = O(1),
	\end{equation*}
	since $k=\left \lfloor{px}\right \rfloor$ and $x\in(0,1)$. Therefore, 
	\begin{equation*}
	\frac{\text{df}_C(k)}{\text{hdf}(k)} = 1 + O\left(\frac{1}{\sqrt{p}}\right) + O\left(\frac{\log(p)}{p} \right),
	\end{equation*}
	and hence
	\begin{equation*}
	\frac{\text{df}_C(k)}{\text{hdf}(k)} \to 1.
	\end{equation*}
\end{proof}

\subsection{Expected KL-based optimism, in the context of BS }
\label{sec:expectedkl_bs}
In this section, we obtain the expected Kullback-Leibler (KL) based optimism for BS with subset size $k$. Let's first consider fitting least squares regression on $k$ prefixed predictors. Recall that 
\begin{equation*}
y = \mu + \epsilon,
\end{equation*}
where $\epsilon \sim \mathcal{N}(0,\sigma^2 I)$. We use the deviance to measure the predictive error, that is 
\begin{equation*}
\Theta=-2 \log f(y|\mu,\sigma^2).
\end{equation*}
The training error is then 
$$\text{err}_{\text{KL}} = -2 \log f (y|\hat{\mu},\hat{\sigma}^2),$$
and the testing error (KL information) is
$$\text{Err}_{\text{KL}}  = -2 E_0 \left[ \log f(y^0|\hat{\mu},\hat{\sigma}^2)\right],$$
where $\hat{\mu}$ and $\hat{\sigma}^2$ are the maximum likelihood estimators (MLE) based on training data $(X,y)$, $y^0$ is independent and has the same distribution of $y$ and $E_0$ is the expectation over $y^0$. 

Due to the assumption of normality, the deviance can be expressed as
\begin{equation}
\Theta = n\log(2\pi \sigma^2) + \frac{\lVert y- \mu \rVert_2^2}{\sigma^2}.
\label{eq:deviance}
\end{equation}
Maximizing the likelihood, or minimizing the deviance \eqref{eq:deviance}, gives
\begin{equation}
\begin{aligned}
& \hat{\mu} = \argmin_\mu  \lVert y-\mu \rVert_2^2,\\
&\hat{\sigma}^2 = \frac{1}{n} \lVert y-\hat{\mu}\rVert_2^2.
\label{eq:appen_mle}
\end{aligned}
\end{equation}
Using these expressions, we then have
\begin{equation}
\text{err}_\text{KL} = n \log(2\pi \hat{\sigma}^2) +n,
\label{eq:err_kl}
\end{equation}
and
\begin{equation*}
\text{Err}_\text{KL} = n\log(2\pi \hat{\sigma}^2) + n\frac{\sigma^2}{\hat{\sigma}^2} +\frac{\lVert \mu- \hat{\mu} \rVert_2^2}{\hat{\sigma}^2}.
\end{equation*}
The expected optimism is then
\begin{equation}
\begin{aligned}
E(\text{op}_\text{KL})  &= E(\text{Err}_\text{KL}) - E(\text{err}_\text{KL}),\\
&= E\left(n\frac{\sigma^2}{\hat{\sigma}^2}\right) + E\left(\frac{\lVert \mu- \hat{\mu}) \rVert_2^2}{\hat{\sigma}^2}\right) -n.
\end{aligned}
\label{eq:eop}
\end{equation}

So far we've been considering a subset with $k$ fixed predictors. At subset size $k$, BS chooses the one with minimum residual sum of squares (RSS) from all $\binom{p}{k}$ possible subsets. In order for the above derivation to continue to hold for BS of subset size $k$, we need to show that the MLE from \eqref{eq:appen_mle} is also the BS fit. This can be easily obtained from the full likelihood (-2 times) \eqref{eq:err_kl}, which after substituting the expression of $\hat{\sigma}$ leads to
\begin{equation*}
n\log\left(\frac{2\pi}{n}\lVert y-\hat{\mu}\rVert_2^2\right) + n.
\end{equation*}
Therefore, for all $\binom{p}{k}$ models of size $k$, the one with largest log likelihood, is also the one with smallest RSS. Hence \eqref{eq:eop} holds for BS fit with subset size $k$ as well.

\subsection{Proof of Theorem \ref{thm:correspondence}}
\label{sec:correspondence}

\begin{proof}	
	Since $[X_1,X_2,\cdots,X_j]$ and $[Q_1,Q_2,\cdots,Q_j]$ span the same space, we have
	\begin{equation}
	\hat{\alpha}^{(j)} = \hat{\beta}^{(j)}.
	\label{eq:thmproof-correspondence-zrq-zrx-subset}
	\end{equation}
	We can express $\hat{\gamma}(k_Q)$ as
	\begin{equation}
		\hat{\gamma}(k_Q) = \sum_{j\in S_k} \hat{\gamma}^{(j)} - \hat{\gamma}^{(j-1)}.
		\label{eq:zs_expand}
	\end{equation}
	We multiply both sides by $R^{-1}$ ($X$ is assumed to have full column rank), and use \eqref{eq:thmproof-correspondence-zrq-zrx-subset} to get
	\begin{equation*}
		\hat{\beta}(k_Q) = \sum_{j\in S_k} \hat{\alpha}^{(j)} - \hat{\alpha}^{(j-1)}.
	\end{equation*}
\end{proof}

\section{Simulation setup}
\subsection{Setup for studying the performance of BS}
\label{sec:simulation_setup_orthx}

We consider a trigonometric configuration of $X$ that is studied by \citet{Hurvich1991}, where $X=(x^1,x^2)$ is an $n$ by $p$ matrix with components defined by 
$$ x_{tj}^1 = \sin\left(\frac{2\pi j}{n}t\right),$$
and 
$$ x_{tj}^2 = \cos\left(\frac{2\pi j}{n}t\right),$$
for $j=1,\cdots,p/2$ and $t=0,\cdots,n-1$. The columns of $X$ are then standardized to have $l_2$ norm 1, to make them orthonormal. By fixing $X$, the responses are generated by \eqref{eq:truemodel_def}, where $\mu=X\beta$. The error $\epsilon$ is also shifted to have mean $0$, hence the intercept will be zero. 

We consider the following configurations of the experiment:
\begin{itemize}
	\item Sample size: $n \in \{200, 2000\}$.
	\item Number of predictors: $p \in \{14,30,60,180\}$.
	\item Signal-to-noise ratio: $\text{SNR} \in \{0.2,1.5,7\}$ (low, medium and high). The average oracle $R^2$ (linear regression on the set of true predictors) corresponding to these three SNR values are roughly $20\%$, $50\%$ and $90\%$, respectively.
	\item Coefficient vector $\beta$ (Orth in the following denotes for orthogonal $X$):
	\begin{itemize}
		\item Orth-Sparse-Ex1: $\beta=[1_6,0_{p-6}]^T$
		\item Orth-Sparse-Ex2: $\beta=[1,-1,5,-5,10,-10,0_{p-6}]^T$
		\item Orth-Dense \citep{Taddy2017}: $\beta_j = (-1)^j \exp(-\frac{j}{\kappa})$, $j=1,\cdots,p$. $\kappa=10$ 
	\end{itemize}
\end{itemize}

In total, there are $72$ different scenarios in the experiment. The full set of simulation results is presented in the Online Supplemental Material. In each scenario, $1000$ replications of the response $y$ are generated. A fitting procedure $\hat{\mu}$, is evaluated via the average RMSE, where 
\begin{equation*}
\text{RMSE}(\hat{\mu}) = \sqrt{ \frac{1}{n} \lVert \hat{\mu}-X\beta \rVert_2^2}.
\end{equation*}
To make the scales easier to compare, we construct two relative metrics: $\%$ worse than the best possible BS, and relative efficiency, which are defined as follows:
\begin{itemize}
	\item \textbf{$\%$ worse than best possible BS}
	\begin{equation*} 
	=\displaystyle 100 \times \left( \frac{\text{averge RMSE of a fitting procedure } \hat{\mu}}{\text{average RMSE of the best possible BS}} - 1 \right) \%,
	\end{equation*}
	where the best possible BS here means that on a single fit, choosing the subset size with the minimum RMSE among all $p+1$ candidates, as if an oracle tells us the best model.
	
	\item \textbf{Relative efficiency:} For a collection of fitting procedures, the relative efficiency for a particular procedure $j$, is defined as
	\begin{equation*}
	\displaystyle \frac{\min_l \text{ average RMSE of fitting procedure }l}{\text{average RMSE of fitting procedure }j}.
	\end{equation*}
	The relative efficiency is a measure between $0$ and $1$. Higher value indicates better performance. Besides the fitting procedures specified, we include the null and full OLS in the calculation of relative efficiency. 
\end{itemize}
We also present the sparsistency (number of true positives) and number of extra predictors (number of false positives).

\subsection{Setup for studying the performance of BOSS}
\label{sec:simulation_setup_generalx}
We consider a general $X$, where  $x_i\sim \mathcal{N}(0,\Sigma)$, $i=1,\cdots,n$ are independent realizations from a $p$-dimensional multivariate normal distribution with mean zero and covariance matrix $\Sigma=(\sigma_{ij})$. 

The correlation structure and true coefficient vector $\beta$ include the following scenarios:
\begin{itemize}
	\item Sparse-Ex1: \textbf{All of the predictors (both signal and noise) are correlated.} We take $\sigma_{i,j}=\rho^{|i-j|}$ for $i,j\in\{1,\cdots,p\}\times\{1,\cdots,p\}$. As to $\beta$, we have $\beta_j=1$ for $p_0$ equispaced values and $0$ everywhere else. 
	\item Sparse-Ex2: \textbf{Signal predictors are pairwise correlated with opposite effects.} We take $\sigma_{i,j}=\sigma_{j,i}=\rho$ for $1\le i <j \le p_0$. Other off-diagonal elements in $\Sigma$ are zero. For the true coefficient vector, we have $\beta_{2j-1}=1$ and $\beta_{2j}=-1$ for $1\le j \le p_0/2$, and all other $\beta_j=0$ for $j=p_0+1,\cdots,p$.
	\item Sparse-Ex3: \textbf{Signal predictors are pairwise correlated with noise predictors.} We take $\sigma_{i,j}=\sigma_{j,i}=\rho$ for $1\le i \le p_0$ and $j=p_0+i$. Other off-diagonal elements in $\Sigma$ are zero. $\beta=[1_{p_0},0_{p-p_0}]^T$.
	\item Sparse-Ex4: \textbf{Same correlation structure as Sparse-Ex2, but with varying strengths of coefficients.} We have $\beta_j=-\beta_{j+1}$ where $j=2k+1$ and $k=0,1,\cdots,p_0/2-1$. Suppose that $\beta^\prime=[1,5,10]$, then $\beta_j=\beta^\prime_k$ where $k=j (\text{mod} 3)$. 
	\item Dense: \textbf{Same correlation structure as Ex1, but with diminishing strengths of coefficients}. The true coefficient vector has: $\beta_j = \displaystyle (-1)^j \exp(-\frac{j}{\kappa})$, $j=1,\cdots,p$, and here $\kappa=10$.
\end{itemize}
The setup of Sparse-Ex1 is very common in the literature, such as in \citet{Bertsimas2016} and \citet{Hastie2017}. All of the predictors are correlated (when $\rho \ne 0$) where the strength of correlation depends on the physical positions of variables. Sparse-Ex2 is designed such that the pair of correlated predictors, e.g. $(X_1,X_2)$, leads to a good fit (high $R^2$), while either single one of them contribute little to the fitted $R^2$. Sparse-Ex4 is similar to Sparse-Ex2, but has varying strengths of coefficients for the true predictors. In Sparse-Ex3, signal predictors are only correlated with the noise ones. Finally, the dense setup is built on the dense example in Section \ref{sec:simulation_setup_orthx}, by having correlated predictors.

For the sparse examples, we take $p_0=6$. We consider three values of the correlation parameter, $\rho \in [0, 0.5, 0.9]$. For the case of $n>p$, other configuration options, including $n$, $p$, and SNR, are the same as in Section \ref{sec:simulation_setup_orthx}. For the case of $n \le p$, we consider $n \in \{200, 500\}$ and $p \in \{550, 1000\}$. This implies a total of $540$ different combinations of configuration options. For each configuration, $1000$ replications are estimated and we present the same evaluation measures as introduced in Section \ref{sec:simulation_setup_orthx}. The full set of results can be found in the Online Supplemental Material.

\section{Extra figures and tables}
\begin{figure}[ht!]
	\centering
	\includegraphics[width=0.9\textwidth]{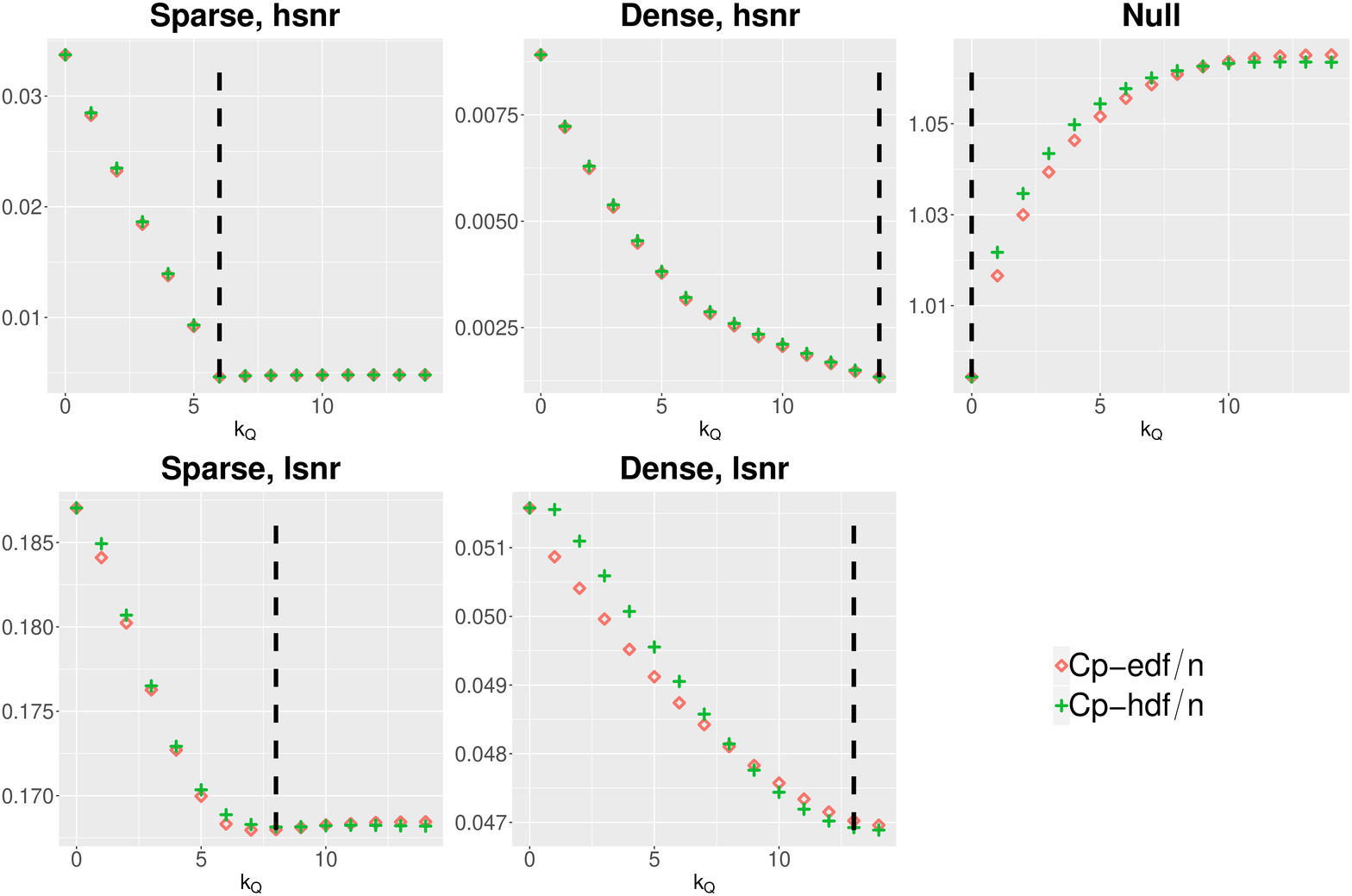}
	\caption{Averages of C$_p$-edf and C$_p$-hdf for BOSS over $1000$ replications. Here $X$ is general with $n=200$, $p=14$. Both criteria result in the same average of the selected subset size over the $1000$ replications (rounded to the nearest integer) as denoted by the dashed vertical lines. We assume knowledge of $\mu$ and $\sigma$.}
	\label{fig:boss_cp_edf_hdf}
\end{figure}

\begin{table}[ht]
\centering
\caption{The performance of BS and LBS. The selection rule for both methods is C$_p$-edf. We assume knowledge of $\mu$ and $\sigma$.} 
\label{tab:lbs_bs}
\scalebox{0.9}{
\begin{tabular}{|c|c|c|cc|cc|cc|}
  \toprule 
 \multicolumn{1}{|c}{} & \multicolumn{1}{c}{} &       & \multicolumn{2}{c|}{Orth-Sparse-Ex1} & \multicolumn{2}{c|}{Orth-Sparse-Ex2} & \multicolumn{2}{c|}{Orth-Dense} \\
 \cmidrule{4-9}\multicolumn{1}{|c}{} & \multicolumn{1}{c}{} &       & BS    & LBS   & BS    & LBS   & BS    & LBS  \\
 \midrule 
 \multicolumn{1}{|c}{} & \multicolumn{1}{c}{} &       & \multicolumn{6}{c|}{\% worse than the best possible BS} \\
 \midrule 
 \multirow{4}[4]{*}{n=200} & \multirow{2}[2]{*}{hsnr} & p=30 & 4 & 28 & 21 & 25 & 1 & 1 \\ 
   &  & p=180 & 1 & 43 & 12 & 25 & 7 & 10 \\ 
  \cmidrule{2-9} & \multirow{2}[2]{*}{lsnr} & p=30 & 20 & 26 & 21 & 30 & 15 & 16 \\ 
   &  & p=180 & 15 & 20 & 15 & 32 & 8 & 12 \\ 
  \midrule \multirow{4}[4]{*}{n=2000} & \multirow{2}[2]{*}{hsnr} & p=30 & 3 & 29 & 3 & 28 & 0 & 0 \\ 
   &  & p=180 & 0 & 42 & 0 & 41 & 6 & 9 \\ 
  \cmidrule{2-9} & \multirow{2}[2]{*}{lsnr} & p=30 & 3 & 28 & 7 & 25 & 2 & 2 \\ 
   &  & p=180 & 0 & 41 & 1 & 37 & 8 & 12 \\ 
   \midrule 
 \multicolumn{1}{|c}{} & \multicolumn{1}{c}{} &       & \multicolumn{6}{c|}{Relative efficiency} \\
 \midrule 
\multirow{4}[4]{*}{n=200} & \multirow{2}[2]{*}{hsnr} & p=30 & 1 & 0.81 & 1 & 0.96 & 1 & 1 \\ 
   &  & p=180 & 1 & 0.71 & 1 & 0.89 & 1 & 0.97 \\ 
  \cmidrule{2-9} & \multirow{2}[2]{*}{lsnr} & p=30 & 1 & 0.96 & 1 & 0.93 & 0.95 & 0.94 \\ 
   &  & p=180 & 1 & 0.96 & 1 & 0.87 & 1 & 0.95 \\ 
  \midrule \multirow{4}[4]{*}{n=2000} & \multirow{2}[2]{*}{hsnr} & p=30 & 1 & 0.8 & 1 & 0.81 & 1 & 1 \\ 
   &  & p=180 & 1 & 0.71 & 1 & 0.71 & 1 & 0.98 \\ 
  \cmidrule{2-9} & \multirow{2}[2]{*}{lsnr} & p=30 & 1 & 0.81 & 1 & 0.86 & 1 & 1 \\ 
   &  & p=180 & 1 & 0.71 & 1 & 0.74 & 1 & 0.97 \\ 
   \midrule 
 \multicolumn{1}{|c}{} & \multicolumn{1}{c}{} &       & \multicolumn{6}{c|}{Sparsistency (number of extra variables)} \\
 \midrule 
\multirow{4}[4]{*}{n=200} & \multirow{2}[2]{*}{hsnr} & p=30 & 6(0.1) & 6(0.7) & 4.8(0.4) & 5(0.9) & 30 & 29.9 \\ 
   &  & p=180 & 6(0) & 6(0.9) & 4.2(0) & 4.6(0.9) & 20.5 & 21.9 \\ 
  \cmidrule{2-9} & \multirow{2}[2]{*}{lsnr} & p=30 & 4.5(1.9) & 4.4(2.3) & 2.7(0.6) & 2.9(1.1) & 12.8 & 7.6 \\ 
   &  & p=180 & 1.9(0.5) & 2.3(1.4) & 1.9(0.3) & 2.2(1.2) & 0.8 & 2.7 \\ 
  \midrule \multirow{4}[4]{*}{n=2000} & \multirow{2}[2]{*}{hsnr} & p=30 & 6(0.1) & 6(0.7) & 6(0.1) & 6(0.7) & 30 & 30 \\ 
   &  & p=180 & 6(0) & 6(0.8) & 6(0) & 6(0.8) & 32.1 & 34.1 \\ 
  \cmidrule{2-9} & \multirow{2}[2]{*}{lsnr} & p=30 & 6(0.1) & 6(0.7) & 4.1(0.2) & 4.3(0.7) & 28.8 & 29.4 \\ 
   &  & p=180 & 6(0) & 6(0.8) & 4(0) & 4.1(0.8) & 13.9 & 15.9 \\ 
   \bottomrule 
\end{tabular}
}
\end{table}

\clearpage
\begin{figure}[!ht]
	\centering
	\includegraphics[width=0.9\textwidth]{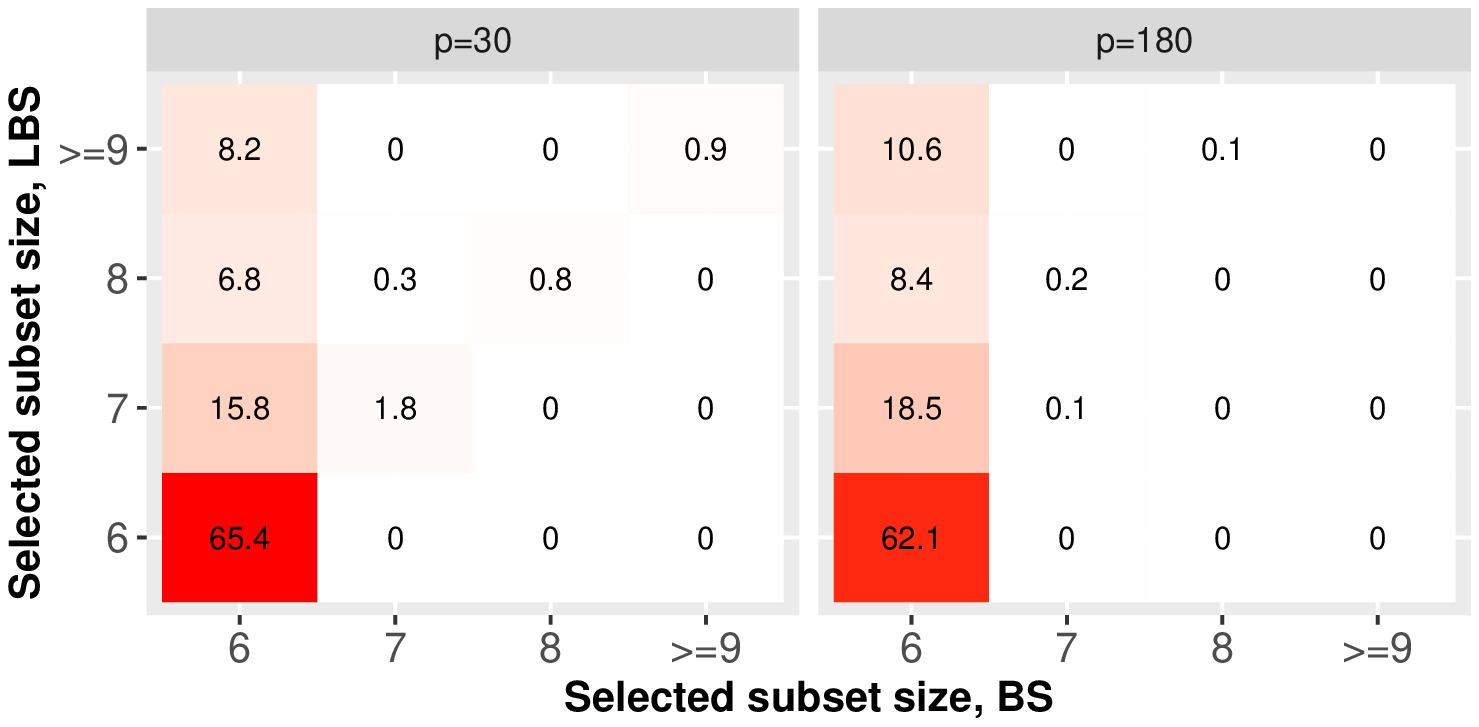}
	\caption{Frequency distributions (in $\%$) of the selected subset size  given by BS and LBS, based on $1000$ replications. The selection rule is C$_p$-edf. The true model is Orth-Sparse-Ex1 with $n=200$, $p_0=6$ and high SNR.}
	\label{fig:numvar_bs_lbs} 
\end{figure}

\clearpage
\bibliographystyleonline{chicago}
\bibliographyonline{reference.bib}
\end{document}